%% file: paper.tex
\documentclass{article}
\usepackage{graphicx} % Required for inserting images
\usepackage{color}
\usepackage[margin=1in]{geometry}
\usepackage{fluff}
\usepackage{comment}

\usepackage{mathtools, xfrac}
\usepackage{hyperref}
\usepackage{pgfplots}
\pgfplotsset{width=10cm,compat=1.9}
\usepackage{subcaption}
\usepackage{booktabs}
\usepackage{makecell, float}
\usepackage{comment}
\usepackage[suppress]{color-edits}
\addauthor{er}{orange}
\addauthor{jk}{green}
\addauthor{so}{blue}
\addauthor{et}{purple}

\newcommand\numberthis{\addtocounter{equation}{1}\tag{\theequation}}

\graphicspath{{figures/}}

\title{Modeling reputation-based behavioral biases in school choice}
\author{
    Jon Kleinberg\\%\thanks{??? Supported in part by a Vannevar Bush Faculty Fellowship, MURI grant W911NF-19-0217, AFOSR grant FA9550-19-1-0183, a Simons Collaboration grant, and a grant from the MacArthur Foundation.}\\
    Cornell University\\
    \texttt{kleinberg@cornell.edu}
    \and
    Sigal Oren\\%\thanks{Supported in part by BSF grant 2018206 and ISF grant 2167/19.}\\
    Ben-Gurion University\\
    \texttt{sigal3@gmail.com}
    \and
    Emily Ryu\\%\thanks{Supported in part by the National Science Foundation Graduate Research Fellowship Program under Grant No. DGE-2139899.}\\
    Cornell University\\
    \texttt{eryu@cs.cornell.edu}
    \and
    \'Eva Tardos\\%\thanks{??? Supported in part by NSF grants CCF-1408673, CCF-1563714 and AFOSR grant FA9550-19-1-0183.}\\
    Cornell University\\
    \texttt{eva.tardos@cornell.edu}
}

\begin{document}

\maketitle
\thispagestyle{empty}
\begin{abstract}
    \input{paper_sections/abstract}
\end{abstract}
\newpage
\setcounter{page}{1}
\input{paper_sections/intro}
\input{paper_sections/motivation}

\input{paper_sections/model}

\input{paper_sections/singleschool}
\input{paper_sections/multschools}

\input{paper_sections/conclusion}

\section*{Acknowledgments}
We thank Kenny Peng and Nikhil Garg for useful discussions of modeling school choice and strategic application behavior. This work was supported in part by BSF grant 2018206, ISF grant 2167/19, the NSF Graduate Research Fellowship Program under Grant No. DGE-2139899, AFOSR grant FA9550-23-1-0410, AFOSR grant FA9550-231-0068, a Vannevar Bush Faculty Fellowship, a Simons Collaboration grant, and a grant from the MacArthur Foundation.

\bibliographystyle{alphaurl}
\bibliography{paperbib}

\appendix
\input{paper_sections/appendix}

\end{document}

%% file: paper_sections/abstract.tex
A fundamental component in the growing theoretical literature on school choice is the problem a student faces in deciding which schools to apply to. Recent models have considered a setting with a set of schools of different selectiveness, and a student who is unsure of their strength as an applicant and can apply to at most $k$ schools. Such models assume that the student cares solely about maximizing the quality of the school that they will attend. However, experience suggests that students' decisions are additionally influenced by a set of crucial behavioral biases based on reputational effects: they experience a subjective reputational benefit when they are admitted to a selective school, whether or not they attend; and a subjective loss based on disappointment when they are rejected. Guided by these observations, and inspired by recent behavioral economics work on loss aversion relative to expectations, we propose a behavioral model by which a student chooses schools in a way that balances these subjective behavioral effects with the quality of the school they eventually attend. 

Our main results show that a student's choices change in interesting and dramatic ways in a model where these reputation-based behavioral biases are taken into account. In particular, where a rational applicant spreads their applications evenly across the spectrum of school selectiveness at optimality, a biased student applies very sparsely to highly selective schools, such that above a certain threshold they apply to only an absolute constant number of schools even as their budget of available applications grows to infinity. Consequently, a biased student underperforms a rational student even when the rational student is restricted to a sufficiently large upper bound on applications and the biased student can apply to arbitrarily many. Our analysis shows that the reputation-based model is rich enough to cover a range of different ways that biased students cope with fear of rejection through their application decisions, including not just targeting less selective schools, but also occasionally applying to schools that are too selective, compared to rational students.

%% file: paper_sections/intro.tex
\newcommand{\xhdr}[1]{\paragraph{\bf {#1}.}}
\newcommand{\omt}[1]{}

\section{Introduction}

In theoretical frameworks for college admissions or job search, 
a fundamental component is the model by which candidates choose
where to apply.
There are a number of distinct considerations that need to go into such a model.
Focusing on the language of admissions, there is a range of colleges
with different levels of selectivity, and the student applying is
unsure of their exact strength as an applicant, so 
they need to diversify where they apply in order 
to optimize the best school they
get into with a limited number of applications.
But there is also a behavioral component, which empirical work has
shown to produce large effects in this type of decision-making with
uncertain accept/reject outcomes: 
the student wants to proactively avoid situations where they
anchor their expectations on a good school that then rejects them,
since this produces large disutility through the interplay
of anticipation and subsequent disappointment.
Although the terminology is not perfectly apt, we will think
of the first of these considerations --- optimizing the quality of
the school you actually attend --- as the {\em rational}
part of the student's utility (since it is just about 
the admissions outcome without considering how the
student experiences it) and the second of these as the 
{\em behavioral} part.

A feature of most college application or job application processes
is that the applications are sent out (at least approximately)
in a single batch, and later
the candidate learns the outcome of all the applications and chooses
among the places where they are accepted.
This is in contrast to processes based on centralized matching
using mechanisms like the Deferred Acceptance algorithm (such as
in high school admissions in some municipalities, or medical resident
matching), where
candidates submit ranked lists and a global authority performs the matching.

In this paper we will consider the single-batch applications process and study the effect of behavioral bias in this process. We will use a basic model studied by Ali and Shorrer 
\cite{as2023college}:
applicants have {\em strengths} $s \in [0,1]$ indicating
the quality of their application. Schools admit based on this application strength, which is known to all schools but not to the applicant. We can label schools by real numbers $x$ in $[0,1]$; 
higher numbers $x$ correspond to better, more selective schools in that
the school with label $x$ only
accepts students of strength at least $x$.
A central question in this model is the following \emph{application portfolio problem}:
if a student can only send out $k$ applications, and 
wants to maximize their expected utility, where should they apply?  Ali and
Shorrer \cite{as2023college} show that the optimal portfolio of $k$ applications can be computed efficiently via dynamic programming. % \socomment{Does it makes sense to explicitly say that in their model there is a set of schools? so differentiating this from the continuous model.  }\etcomment{I was think that this would be too low level detail at this point. Maybe better at related work. you would like to see more detail here?}\jkcomment{I agree; in all cases there's an input set of schools, and the distinction between finite sets and a continuum doesn't seem crucial for this paragraph.} \ercomment{also agree} \socomment{I think the distinction is already built in the text above when we say they solve it using a dynamic program. This confused me since at first I thought they are also using continuum. Also, in sentence two of the next paragraph we do say that we assume a continuum so in some sense it is a distinction that we are making here. If you don't think other people will be confused about this we can leave it. }

Throughout this paper we will use a structured version of this model that allows us to expose the basic features and
surprising phenomena caused by the behavioral effects we are modeling.  We assume that there is a continuum of schools labeled by the real numbers in $[0,1]$; the school with label $x$ (which, as before, only accepts students
of strength at least $x$) has a utility of $x$ to the student, 
and the student's uncertainty about their application is reflected in the fact that their strength is drawn uniformly at random from $[0,1]$.
So from the student's perspective, the probability that they will be accepted to school $x$ is $1 - x$, and these acceptance events are dependent, in that if they get into school $x$, they will also get into school $x' < x$. Ali and Shorrer provided a solution
to this application portfolio problem in this special case for a rational agent, whose
expected utility is simply the strongest school they are admitted to;
roughly speaking, 
at optimality the student should evenly space their applications on the unit
interval, applying to schools $1/(k+1), 2/(k+1), \ldots, k/(k+1)$
and thereby covering the full interval increasingly densely, 
and uniformly, as $k$ increases \cite{as2023college}.

\xhdr{Our model: Incorporating behavioral considerations}
In this paper, we ask how these rational strategies compare to a model
in which an applicant also experiences  behavioral effects. For this, we draw on a well-established
approach to modeling behavioral utilities in 
a two-phase process such as this, where an individual seeks out valuable
options in a first phase and learns their outcomes in a second phase.
This approach, known as {\em expectations-based reference-dependent
preferences (EBRD)} \cite{KR06, KR07, KR09},
has been used in behavioral models for the different
context of Deferred Acceptance processes
\cite{dreyfussEBRD}, a paper that helped motivate our questions here.

% Jon: Trying some edits to the behavioral motivation based on our call

In our setting, we build a model based on these principles as follows.
First, as in the rational case, a student will derive a utility of
$x$ if they ultimately are accepted to, and select, school $x$.
They will only experience this part of the utility from the school
they actually end up attending, since it is based on actually
consuming the opportunity they are offered, and they can only attend
one school.
But beyond this,
a student also experiences a behavioral component in their utility 
from each school they apply to:
the process of applying creates an anchor point based on the 
expected value of the school, and
the student then receives either positive or negative subjective
value from their outcome relative to this anchor point, because
they either overperformed or underperformed their expectation.
%\jkreplace{We will refer to a student who {exhibits behavioral bias in her choices} as a {\em biased student}.}
{As we will see next, these behavioral considerations
will generally lead to bias in the decision-making, and we refer to
a student who incorporates them as a {\em biased student}.} 
% \socomment{As we discussed with $\gamma=1$ the agents behave completely rationally. Should we defer the definition of bias agents till after we add the loss aversion component? }\etcomment{I don't think deferring the reference to biased, I am was Ok with the way the text was. Maybe even preferred that. But also suggested an alternative avoiding the issue Sigal is raising.}\jkcomment{I prefer referring to bias at this point, and tried a further edit of Eva's sentence, above.} \ercomment{i like this} \socomment{The new sentence sounds good to me}

Here is how the model of this behavioral effect works in more detail.
Let's consider a school with admission probability $p$ and 
utility $v$ from attendance.  
(In our case, schools have $v = x$ and $p = 1-x$ for some $x \in [0,1]$,
but it is useful to consider the definition with $p$ and $v$
defined arbitrarily.)
We build up the behavioral component of the utility in the following steps:
\begin{itemize}
\item 
The basic ``units'' for this part of the utility are scaled by
a coefficient $\tau$, which gives the strength of the behavioral effect.
For a school that confers consumption utility $v$ from attending, it also has
a {\em subjective value} of $\tau v$ when considering the behavioral impact
of being admitted or rejected.
\item 
The behavioral effect is realized in two stages.
When the applicant first applies to the school, 
they anchor on the expected value $p \tau v$ as a default reference point
for the subjective utility they aspire to receive from this application.
(This anchoring gives the general formalism its name in the literature:
it is a reference-dependent preference that is based on expectations.)
\item 
In a second stage, the applicant learns whether or not they are admitted.
If they are admitted, then the realized gain relative to this reference
point, which is $\tau v - p \tau v = (1-p) \tau v$. This is the subjective
benefit they experience from being accepted.
If they are rejected, then they experience a corresponding subjective loss
of $p \tau v$ relative to the value they anchored on.
\item
If this were the entire process, then an applicant would break even
in expected utility on all the schools they apply to --- anchoring on
the expected value $p \tau v$ in advance, 
and then either realizing the {additional part of the} full subjective value $\tau v$ if accepted, or 
paying back what they anchored on if they are rejected. % \etcomment{I think we should avoid the word "borrowing" as that more suggest that this $ p \tau v$ will be part of their value. } \socomment{I agree}
But a line of empirical work establishes that the ubiquitous behavioral
mechanism of loss aversion is at work in this process, and 
the loss in the case of rejection is magnified
\cite{artemov2017strategic,hassidim2021limits}.
Thus, we say that when the applicant is rejected, they experience a
disutility that is not $- p \tau v$ but instead $- \lambda p \tau v$
for a value $\lambda > 1$ --- an intuitively familiar experience in
which the disappointment from rejection is amplified relative to
the amounts involved.
\end{itemize}
So this gives the full form of the student's expected behavioral utility 
from applying to a given school with acceptance probability $p$ and value $v$:
relative to the reference point of $p \tau v$ they set at the time
they apply, they gain $(1-p) \tau v$ with probability $p$, and they
lose $\lambda p \tau v$ with probability $1 - p$, for a total of
$$p(1-p) \tau v - (1-p) \lambda p \tau v 
 = (1 - \lambda) \tau (1-p)pv = - \gamma (1-p)pv,$$
where in the last equality we have defined a new coefficient
$\gamma = (\lambda - 1) \tau > 0$
to write the expected behavioral utility more compactly.
Note that this expression is negative for $\lambda > 1$, and this
creates an aversion for applicants to apply based on the fear of rejection.
We will refer to $\gamma$ as the student's {\em bias parameter} and say that a student with $\gamma>0$ is a $\gamma$-biased student;
we expect students with small values of $\gamma$ to behave more similarly
to a rational student, since $\gamma = 0$ corresponds to the case 
in which the behavioral effects do not shift the utility calculations.
In our subsequent discussion of related work, we go into more detail
on how this formalism relates to other uses of the EBRD framework 
for problems in which agents anchor on expected gains with
the possibility of rejection.

We'll assume that the biased student still chooses to attend
the highest-quality school among the ones where they are admitted.
If we define the {\em admissions outcome} for a student as the quality of this
school that they ultimately attend, then for a biased student,
we have a gap between their expected admissions
outcome --- which is based entirely
on the quality of the school they attend --- and their expected utility,
which consists of the admissions outcome plus 
the behavioral terms as well.
In contrast, for our notion of a rational student, who does not experience
the behavioral effects, the admissions outcome and the utility are the same.

\xhdr{Our results: The effect of behavioral biases}
We therefore come to the basic question that motivates the paper:
how does a biased student solve the application
portfolio problem, and how does their expected admissions outcome ---
the quality of the school they eventually attend --- 
compare to the corresponding outcome for a rational student
as a function of their bias parameter\etdelete{ $\gamma$}?

We find that the biased student's solution to the application portfolio
problem is dramatically different from the rational student's,
for every positive bias parameter $\gamma$.
In contrast to the rational student's strategy of using their $k$
applications to apply uniformly across the range of all schools,
the biased applicant applies very sparsely to more selective schools,
and focuses their applications increasingly densely on less 
selective schools (see Section~\ref{sec:undershooting}). 
This effect becomes extreme in the limit as $k$ grows ---
even as $k$ goes to infinity (when the rational agent's applications
produce a dense cover of the unit interval), 
we show that the number of schools that the biased agent applies to above any
constant $c \in [0,1]$ remains bounded by an absolute constant
independent of $k$.
(As an example of this effect instantiated with concrete parameters, see Figure \ref{fig:portfolio_numberline}, which shows the portfolio of schools chosen by a student with bias parameter $\gamma = 0.1$ who sends out $k = 100$ applications.)
We prove other versions of this as well; for example, for every
$\gamma > 0$ there is a number $h(\gamma) < 1$ such that a 
student with bias $\gamma$
applies to at most one school above $h(\gamma)$ even as
$k$ goes to infinity.

\begin{figure}[H]
    \centering
    \includegraphics[width=0.8\linewidth]{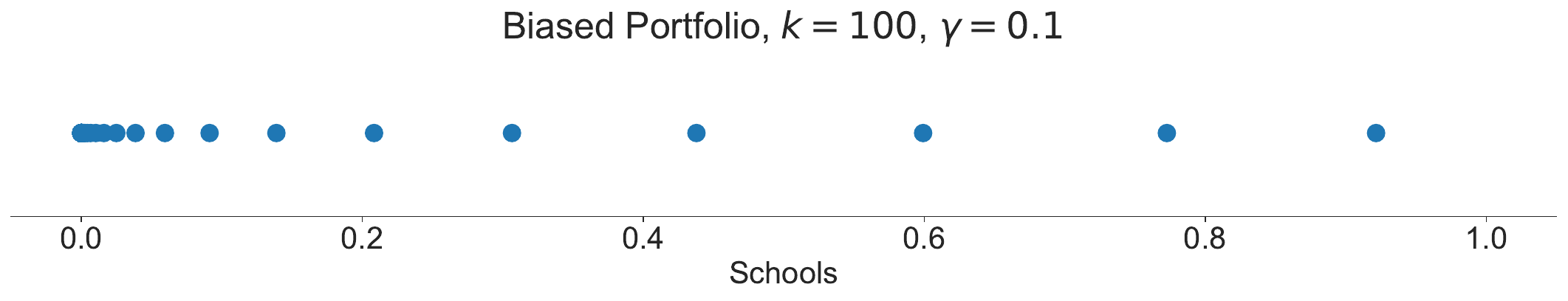}
    \caption{As shown here, a biased student in our model concentrates most of their applications on less selective schools, and applies very sparsely to the more selective end of the range.  The plot shows the positions (on the interval $[0,1]$) of $k = 100$ applications sent out by a student with bias parameter $\gamma = 0.1$.}
    \label{fig:portfolio_numberline}
\end{figure}

A consequence of this sparsity at the top end of the school selectivity
range is the following result: 
a biased student sending an unbounded number of applications will
have a lower expected admissions outcome than a rational student
who sends a sufficiently large constant number of applications.
More precisely (as we also show in Section ~\ref{sec:undershooting}),
for every $\gamma > 0$, there is a constant $k(\gamma)$ such that
a student with bias $\gamma$ and any number of applications has a lower
expected admissions outcome than a rational student 
who sends $k(\gamma)$ applications.
This is a very strong type of gap between rational and biased behavior;
it is not the case that in order to match the outcome of 
a rational student with $k$ applications, a biased student needs some
$f(k)$ applications, but instead that there are some values of $k$
where the biased student simply won't be able to
match the expected outcome --- given their strategy ---
no matter how many schools they are allowed to apply to.

Finally, we prove a further set of results establishing that
our model of behavioral biases is rich enough to cover
qualitatively different ways in which a biased student can cope
with the possibility of rejection (Section~\ref{sec:mon}).
In particular, if an applicant is concerned about rejection, there are two
ways of reducing the anticipated value they anchor on --- 
either by applying to a school that is less selective (thereby
reducing the realized benefit if admitted) or by applying to a school 
that is too selective (thereby reducing the expectation by making
admission more unlikely).
For different parameter ranges, our model produces both behaviors:
not only do biased students over-concentrate on less selective schools,
but for small values of $\gamma$ they also engage in a surprising
type of ``overshooting'' behavior, in which the top schools they apply
to are too selective, allowing them to reduce the expectation they
anchor on because the school's admission probability is so low.
This range of different behaviors illustrates the richness of
the model, and the ways in which it captures phenomena that are
familiar from our everyday experience of applying and experiencing rejection.
% \socomment{Do we want to add an outline for the paper?}\etcomment{I vote for "no". I think with catchy section titles, the outline is not adding anything useful. Maybe we can add section references to the list of our results. that takes less space and at least as useful in my experience}\jkcomment{I like the idea of adding parenthetical section references when we talk about the results.  I didn't do that yet, but could add them.} \ercomment{i took a stab at this, but realized that most of the results are just all in Section 5 (and we don't really have subsections). would it be helpful to add direct Theorem pointers instead / in addition?} \jkcomment{I think the section numbers are the right granularity; I made some edits to this.}

% Earlier text from Sigal, incorporated into the above:
\omt{
Often the utility of students that apply to school consists of two components: 1. A standard consumption utility for a school they were admitted to and decided to attend; 2. a \emph{behavioral} utility that corresponds to a reputation gain from being admitted to a school or a reputation loss from being rejected from a school. 
The utility from the first component equals the maximal value of a school that the student was admitted to. The second utility is exhibited for each one of the schools that the student applies to. The rationale here is that this reptutational term is exhibited when the students get to see her application results or tell her friends about them and in a sense is independent from the results of her other applications.

This view fits well the spirit of a model of expectations-based reference-dependent preferences (EBRD) by K\H{o}szegi and Rabin~\cite{KR06, KR07, KR09}. In the EBRD model, the agent separately derives two forms of utility: the standard ``consumption utility'' and an additional reputational cost termed 'new utility' relative to a \emph{reference point} determined by her expectations about her environment. An agent with EBRD preferences will take into account how her strategy affects her expected outcome, and thus her reference point and gain-loss utility. \socomment{This was copied from Emily's background}

In the classic ERBD model the new cost is related to the consumption. Per the previous discussion, in our model we explicitly distinguish between the two. Let $v$ be the consumption value of some school, then we say that a student that is admitted to the school receives an $\rho \cdot v$ reputational utility for being admitted to the school regardless of whether she attends the school or not. As in the ERBD, reputation utility is relative to a reference point which is defined for each school as the expected reputational gain: $p\cdot \tau \cdot v$. If the student is admitted to the school it gains an extra reputational utility of $(1-p) \tau v$  while if it is rejected it incurs a reputational cost of $p \cdot \tau v$. Note that a rational agent that has these behavioral term behaves exactly the same as an agent that does not have them since the reputaional costs and benefit cancel out and will not affect her expected utility ($p\cdot(1-p) \tau v - (1-p)p\tau v = 0$). However, as anyone that was ever rejected knows rejections are much more potent than success, as in the celebrated line of work on prospect theory \cite{} on loss aversion. Propsect theory tells us that people tend to put more emphasis on their losses than their gains. The same assumption is also applied by the ERBD model. Hence, in this work we consider biased students that have loss aversion. We choose model loss aversion linearly and assume that a biased agent with loss parameter $\lambda>1$ has the following expected reputation utility $p\cdot(1-p) \tau v - \lambda(1-p)p\tau v = -\gamma(1-p)p v$. It will turn out to be convenient to set $\gamma = (1-\lambda)\cdot \tau$ and study the problem in terms of $\gamma$ instead.

\socomment{I think we need another parameter for this to make sense. So assume the reputations cost is $\rho$ times the consumption value. We can get rid of this parameter at the same time we use $\gamma$ instead of $\alpha$ and $\beta$. Also instead of having $\alpha$ and $\beta$ I switched it to $1$ and $\lambda$ which seems somewhat simpler. }
}

%% file: paper_sections/motivation.tex
\section{Background \& motivation} \label{sec:motivation}

% \todo{a couple sentences about matching markets and prefs and DA}\etcomment{Given how much we have in the intro now, I don't think we need this.} \ercomment{sounds good to me, i'll remove this if everyone else also agrees}

\subsection{Simultaneous applications and short lists}
The school choice literature of decentralized and/or simultaneous applications (such as the US college admissions system) has often been used the canonical framework of \emph{simultaneous search} due to Chade and Smith~\cite{chade2006simultaneous}, in which a decision-maker must choose a portfolio of gambles on stochastically independent lotteries. In the language of school choice, the student considers applying to a small subset of $k$ out of $n$ schools, where each school has an exogenously fixed acceptance probability, and admissions decisions are \emph{mutually independent}. In this setting, the optimal $k$-portfolio can be computed efficiently using a greedy algorithm or dynamic programming, and expands upwards (to include increasingly competitive schools) as $k$ increases.

Despite the elegance of the simultaneous search framework, one criticism it has received is that it does not predict the empirically observed strategy of safety schools (i.e. downwards expansion to include less competitive schools as a ``safety net''); this is due to the fact that admissions decisions are generally \emph{correlated} rather than independent. To address this, Ali and Shorrer propose to model correlation via a \emph{threshold model}~\cite{as2023college}. Here, the student has a single ``common score'' that is known to each school (but not to the student herself), there is a finite set of schools, each school has an exogenously fixed acceptance threshold, and students have varying utilities from attending each school. The student now attempts to choose a subset of $k$ schools to apply to, based on her distribution of beliefs over her score. Ali and Shorrer~\cite{as2023college} show that the optimal $k$-portfolio can be computed efficiently via dynamic programming, and demonstrates both upwards and downwards expansion (reflecting the notion of a reach/match/safety strategy). \cite{as2023college} also attempt to expand their findings to an admissions process with partially correlated decisions (i.e. interpolating between their model and that of~\cite{chade2006simultaneous}).
% \etdelete{, but the resulting model becomes quite challenging to analyze mathematically}. \ercomment{this last sentence is an attempt at foreshadowing our justification for also using a linear threshold model for tractability}\etcomment{quite challenging mean here means that  "general analysis of this setup is beyond our scope", that is, hey dont do it. If we want to keep this sentence, I propose something stronger. Like "resulting model becomes quite too complex to offer a general analysis". Maybe. Or do we really want this here?} \ercomment{this is section 5.2 and appendix D.3 of the ali/shorrer paper -- they basically only have 2-3 results, none of which are particularly interesting/insightful, but the math does get quite technically involved so i didn't want to be too dismissive (and thus wasn't sure how to word this). i guess we could just leave it out?} \socomment{I think it is best to leave this out.} \ercomment{seems fair}\etcomment{how is just deleting the end, as I am suggesting above. "attempt to expand" is maybe a fair description suggesting that they do something but not as good as the shared score case.}

Other works such as \cite{haeringer2009constrained} and \cite{gimbert2021constrained} have studied school choice with short lists and portfolio size constraints under various other mechanisms and sources of uncertainty, but their results primarily focus on the existence and structure of equilibria in multiple-student settings, rather than characterizing student strategies or outcomes.

This line of work has focused primarily on studying the behavior of rational students faced with application costs or size constraints; to the best of our knowledge, none has attempted to incorporate a model of behavioral bias or loss aversion into its analysis, as we now propose to do.

\subsection{Behavioral bias in matching markets}
Theoretical work on school choice has traditionally focused on modeling the behavior of rational students, but empirical research has actually found evidence of behavioral biases in how students apply to schools.
Despite the fact that Deferred Acceptance (DA) is strategyproof, numerous studies have repeatedly found empirical evidence of participants misreporting their true preferences, both in field and lab settings. These misreports include ``obvious flippings,'' such as ranking a smaller amount of money higher than a larger amount of money \cite{HKsurvey}, and ``obvious droppings,'' such as ranking an unfunded spot in a program but completely excluding a spot in the same program funded by a fellowship \cite{hassidim2021limits, shorrer2018obvious, artemov2017strategic}. To explain these seemingly obviously dominated misreports, one theory that has been proposed is \emph{behavioral bias}, in which agents appear to play suboptimally (with respect to the classical utility function) because they are actually optimizing for an alternate utility function which incorporates behavioral preferences and/or biases.

{Within this realm, one potential explanation that has received attention is the fear of rejection or disappointment. Dreyfuss et al.~\cite{dreyfussEBRD} proposed to model this by applying a simplified version of the \emph{expectations-based reference-dependent preferences} (EBRD) model of K\H{o}szegi and Rabin \cite{KR06, KR07, KR09} to school choice. In the EBRD model, the agent separately derives two forms of utility: the standard ``consumption utility'' and an additional ``gain-loss utility''(or ``news utility'') relative to a \emph{reference point} determined by her expectations about her environment. An agent with EBRD preferences will take into account how her strategy affects her expected outcome, and thus her reference point and gain-loss utility.}

% \sodelete{Dreyfuss et al.~\cite{dreyfussEBRD} proposed \soreplace{an adaptation~}{implementing a simplified version} of the EBRD model of  K\H{o}szegi and Rabin \cite{KR06, KR07, KR09} to school choice. \soedit{This implementation} 
% One key limitation of the way \cite{dreyfussEBRD} \soreplace{adapted~}{applied} the EBRD model to school choice is that it}
The model of Dreyfuss et al.~\cite{dreyfussEBRD} only considers the utility of a multiple-school application portfolio within the context of DA. While variants of DA are often the mechanism of choice in systems with centralized clearinghouses, many other real-world scenarios (e.g. college applications in the United States) use \emph{decentralized} systems in which participants choose a \emph{set} (rather than a rank-order list) of schools to apply to; this construction thus cannot be applied to understand the effect of loss aversion on students' behavior in such systems. Further, in the model used by \cite{dreyfussEBRD}, the student experiences negative gain-loss utility with respect to the reference point for \emph{every} school she does not attend; as a result, she must pay back her ``borrowed anticipation'' with a penalty of $\lambda > 1$ even for those schools she ranks \emph{lower} than the one she ends up attending.  We propose a new model of reputation-based behavioral bias in school choice, aiming to model fear of rejections (rather than just not attending a school). In our model, the student still evaluates her utility in terms of gain/loss with respect to an anticipated reference point for each school, but we model gaining positive utility from being admitted to schools, even when not attending, and only experiencing negative utility from schools that reject her.

%% file: paper_sections/model.tex
\section{A model of students with reputation-based utility} \label{sec:modelsetup} 

In our model, students obtain both \emph{consumption utility} and \emph{subjective value} (potentially either a benefit or a loss) from schools. Suppose school $i$ is parameterized by its acceptance probability and value to the student: $(p_i, v_i)$. Let $\tau$ be
a parameter measuring the strength of the student's behavioral bias, so that a school that confers \emph{consumption utility} $v$ has a \emph{subjective value} of $\tau v$ that the student considers when weighing the possibility of admission or rejection. At the time the student applies to school $i$, she forms a reference point at the expected subjective value of that school, $p_i \tau v_i$, around which she will compare her realized outcome. If the student is accepted, she receives the additional utility (say, due to a boost in reputation or morale) of the school's full subjective value relative to her anticipated reference point; that is, she gains a subjective benefit of $\tau v_i - p_i \tau v_i = (1-p_i) \tau v_i$. On the other hand, if the student is rejected, she receives disutility from the loss of her anticipation; to model the well-established phenomenon of loss aversion, this disutility is magnified by a coefficient $\lambda > 1$ to give a subjective loss of $-\lambda p_i \tau v_i$.

%If the student is accepted, she receives the additional utility (say, due to a boost in reputation or morale) of the school's value relative to her anticipated reference point\etedit{. Suppose the student weights their reputational utility (relative to consumption utility) by $\alpha$}; that is, she gains \emph{reputational utility} of $\etedit{\alpha(}v_i - p_iv_i\etedit{)} = \etedit{\alpha}(1-p_i) v_i$. On the other hand, if the student is rejected, she receives disutility (due to the pain of rejection) equal to the loss of her anticipation\etedit{. We will use $\beta$ to denote how the student weights her disappointment utility relative to consumption, then }\etdelete{;} her \emph{disappointment utility} is $-\etedit{\beta}p_iv_i$. Finally, the student receives \emph{consumption utility} (in the classical consumption sense) from the single school that she ultimately chooses to attend. \ercomment{need to incorporate $\tau$ into this section too?}

In summary, for each school $i$, the student obtains 
\begin{enumerate}
    \item \emph{Subjective benefit}: $(1-p_i) \tau v_i$ if accepted to $i$.
    \item \emph{Subjective loss}: $-\lambda p_i \tau v_i$ if rejected from $i$.
    \item \emph{Consumption utility}: $v_i$ if the student attends $i$, $0$ otherwise (in other words, $v_i$ if and only if $i$ is the highest ranked school where the student is accepted).
\end{enumerate}
% Further, suppose the student weights their reputational and disappointment utilities (relative to consumption utility) by $\alpha$ and $\beta$, respectively, and define $\gamma \coloneqq \beta - \alpha$. 
For notational simplicity, we define a new coefficient $\gamma = (\lambda - 1) \tau$. Then, for an application portfolio $P$, a \emph{$\gamma$-biased student} has \emph{perceived expected utility} 
\begin{align*}
    U_\gamma(P) &\coloneqq \sum_{i\in P} \left[p_i (1-p_i) \tau v_i - (1-p_i) \lambda p_i \tau v_i + v_i \PP(\text{student attends } i) \right] \\
    &= \sum_{i\in P} (1-\lambda) \tau p_i (1-p_i) v_i + \sum_{i\in P} v_i \PP(\text{student attends } i) \\
    &= -\sum_{i\in P} \gamma p_i (1-p_i) v_i + \sum_{i\in P} v_i \PP(\text{student attends } i). 
\end{align*}

In the third line, we have separated the utility into a \emph{biased} term combining the subjective benefit and loss (analogous to the EBRD news utility) and an \emph{unbiased} consumption utility term.

Lastly, note that we are most interested in the case of $\gamma > 0$ (again analogous to the EBRD model, in that due to loss aversion, the effect of bad news is more significant than the effect of good news), but also attempt to prove statements in the fullest generality possible; each of our results will clearly distinguish the appropriate regime of $\gamma$ for which it holds.

% \begin{remark}
%     \ercomment{is this remark helpful?} \socomment{I'm not sure we should get into this. Also, do you mean that an application sets a reffence point? since this is the whole idea of EBRD.} \etcomment{maybe best to just delete this remark} A criticism of the above explanation might be that it involves a changing reference point. One justification might be that since the utility is separable, it makes sense for the student to also have a \emph{separate} reference point for each school. \etcomment{I would skip the proposed criticism and the part till now. Many start the next sentence with "A way to think about the propose expression is...} Another justification for these expressions could be simpler first principles -- being accepted to school $i$ gives a student more of a ``reputational boost'' if $p_i$ is small or $v_i$ is large (because the school is competitive and/or popular); conversely, rejection causes more disappointment/pain if $p_i$ is large (because the student really expected to get in) or $v_i$ is large (because the student really liked the school), and the expressions reflect this.
% \end{remark}

Throughout this paper we will use a simple model that allows us to expose the basic features and surprising phenomena caused by the bias of the loss aversion we are modeling. 

We will assume that schools lie on a continuum $S = [0,1]$, such that for every $x \in [0,1]$, there exists a school which has both value and acceptance threshold $x$. Further, we will suppose that the student's score is distributed uniformly on $[0,1]$ (so that she is accepted to school $x$ with probability $1-x$).

%% file: paper_sections/singleschool.tex
\section{Preliminaries and Basic Observations } \label{sec:prelim}
\subsection{Warm-up: Applying to a single school} \label{sec:singleschool}

% \etdelete{To illustrate some basic features of our model, first suppose that schools lie on a continuum $S = [0,1]$, such that for every $x_i \in S$, there exists a school which has both value and acceptance threshold $x_i$. Suppose further that the student's score is distributed uniformly on $[0,1]$ (so that she is accepted to school $x_i$ with probability $1-x_i$), and must choose a single school to apply to. 

% Then}
When the student can only apply to a single school, the expected payoff of a student with bias $\gamma$ who applies to school $x$ is 
\begin{equation*}
    U_\gamma(x) = -\gamma (1-x) x^2 + x(1-x) =x(1-x)(1 - \gamma x),
\end{equation*}
which is a cubic with roots at $0$, $1$, and $\frac{1}{\gamma}$. If $0 < \gamma \le 1$, then $\frac{1}{\gamma} \ge 1$, so the graph of the cubic resembles Figure~\ref{fig:singleschool_smallgamma}. Observe that every school $x\in [0,1]$ gives the student non-negative expected payoff, and there exists a unique optimal school $x^* \in (0,1)$.

On the other hand, if $\gamma > 1$, then $0 < \frac{1}{\gamma} < 1$ and the cubic resembles Figure~\ref{fig:singleschool_biggamma}. Now, in addition to the unique optimal school $x^* \in (0,1)$, there also exists a unique \emph{worst} school $\underline{x} \in (0,1)$. In particular, note that $\underline{x} < 1$ with strict inequality, so the worst school is \emph{not} necessarily the most competitive school (because the student assigns a low probability of acceptance to such a school, thus forming a low reference point and experiencing less disappointment upon rejection). Further, all schools with $x > \frac{1}{\gamma}$ give the student negative utility, so the student will only ever apply to schools with $x < \frac{1}{\gamma}$. That is, as $\gamma$ increases, the student places more importance on disappointment utility than reputational utility, and becomes increasingly unwilling to apply to competitive schools for fear of rejection.

In either case, additionally observe that $U_\gamma'(\frac12) = -\frac{\gamma}{4} < 0$, meaning that $x^* < \frac12$. That is, a student with any bias $\gamma > 0$ will shade downwards from the optimum of $\frac12$ (the maximizer of $U(x) = x(1-x)$, the unbiased/true expected payoff), indicating that the fear of rejection causes her to undershoot her true potential.

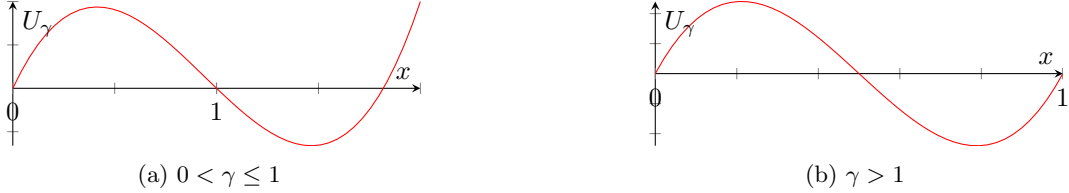
\begin{figure}[H] \label{fig:payoff_single}
   \begin{subfigure}[b]{.48\textwidth}
      \centering
      \begin{tikzpicture}
        \begin{axis}[
            width = 7cm,
            height = 3.5cm,
            axis lines = middle,
            xlabel = $x$,
            ylabel = $U_\gamma$,
            xticklabel = \empty,
            yticklabel = \empty,
            extra x ticks = {0, 1},
            extra x tick style={xticklabel=\pgfmathprintnumber{\tick}}
        ]
        \addplot [
            domain=0:2, 
            samples=50, 
            color=red,
        ]
        {5*x*(1-x)*(1-0.55*x)};
        \end{axis}
        \end{tikzpicture}
      \caption{$0<\gamma \le 1$}
      \label{fig:singleschool_smallgamma}
   \end{subfigure}
   \hfill
   \begin{subfigure}[b]{.48\textwidth}
      \centering
        \begin{tikzpicture}
        \begin{axis}[
            width = 7cm,
            height = 3.5cm,
            axis lines = middle,
            xlabel = $x$,
            ylabel = $U_\gamma$,
            xticklabel = \empty,
            yticklabel = \empty,
            extra x ticks = {0, 1},
            extra x tick style={xticklabel=\pgfmathprintnumber{\tick}}
        ]
        \addplot [
            domain=0:1, 
            samples=50, 
            color=red,
        ]
        {5*x*(1-x)*(1-2*x)};
        \end{axis}
        \end{tikzpicture}
      \caption{$\gamma > 1$}
      \label{fig:singleschool_biggamma}
   \end{subfigure}
   \caption{Expected (biased) utility $U_\gamma(x)$ for a student applying to a single school $x$.}
\end{figure}

%% file: paper_sections/multschools.tex
%\subsection{General multi-school setting: Suboptimal $\gamma$-biased behavior} \label{sec:multschools}
\subsection{Multi-school setting and the optimal portfolio}
\label{sec:multschools}
% \etcomment{previous proposed title: "General multi-school setting - formulation and the unbiased optimal portfolio"} 
% \ercomment{title could probably use some work UPDATE -- yes this one seems better, thanks eva}

%\subsection{Model and the unbiased optimal portfolio}

Now, suppose that the student can simultaneously apply to a subset of schools of size $k$. %\sodelete{(following the lead of~\cite{as2023college}, we will continue the use of the linear threshold model to maintain tractability)} \socomment{I think it is enough to mention this in the intro and perhaps the one school setting.} 
For notational convenience denoting $x_0 = 1$, a student with bias $\gamma$ perceives a portfolio $x_1 > x_2 > \dots x_k$ to have expected value
\begin{align*} 
    U_\gamma(x_1, \dots, x_k) = \sum_{i=1}^{k} \left[ -\gamma (1 - x_i)x_i^2 + x_i (x_{i-1} - x_i) \right], 
\end{align*} 
which follows from the fact that the student is accepted to the school of value $x_i$ if her score is in $[x_i, 1]$ (which occurs with probability $1-x_i$), and attends this school if her score is in $[x_i, x_{i-1})$ (which occurs with probability $x_{i-1}-x_i$).   

\begin{proposition} \label{prop:spacing}
For a portfolio to maximize this perceived expected value, the first order optimality conditions are
\begin{equation} 
    x_{i+1} = 2x_i - x_{i-1} + \gamma(2x_i - 3x_i^2) \label{eq-der-0}
\end{equation}
For an unbiased student (with $\gamma=0$), the resulting
optimal set of $k$ schools in $[0,1]$ is $\frac{i}{k+1}$ for $i\in [k]$.
\end{proposition}
\begin{proof}
The first order  condition for optimality is: 
$$
    \frac{\partial U_\gamma}{\partial x_i} = -\gamma(2x_i - 3x_i^2) + x_{i-1} - 2x_i + x_{i+1} = 0 $$
as claimed. 

For an unbiased agent we get $x_{i+1} = 2x_i - x_{i-1}$. The solution to this system is an arithmetic progression with endpoints at $0$ and $1$. That is, $x_i = \frac{i}{k+1} $ for $i\in [k]$, corresponding to an equally spaced set of schools between $0$ and $1$.
\end{proof}

\begin{corollary} \label{cor:unbiased1/2}
    As $k\to \infty$, the optimal unbiased expected payoff converges to $\frac12$.
\end{corollary}
\begin{proof}
    By Proposition~\ref{prop:spacing}, the optimal set of $k$ schools in $[0,1]$ is $\frac{i}{k+1}$ for $i\in [k]$, which results in an unbiased expected payoff of $$\sum_{i=1}^k \frac{i}{k+1} \cdot \frac{1}{k+1} = \frac{k(k+1)}{2(k+1)^2} = \frac{k}{2(k+1)}.$$ As $k$ grows large, this expression converges to $\frac12$.
\end{proof}

\subsection{Observations on the behavior of a $\gamma$-biased student}

Define $\Delta_i \coloneqq x_{i-1} - x_i$, then observe that (\ref{eq-der-0}) rearranges to
\begin{align*} 
    x_{i-1} - x_{i} - \gamma(2x_i - 3x_i^2) &= x_i - x_{i+1} \\
    \implies \Delta_{i+1} &= \Delta_{i} - \gamma(2x_i - 3x_i^2). \numberthis \label{eq:Deltagap}
\end{align*}

From the observation that $2x_i - 3x_i^2 < 0$ for $x_i > \frac23$ and $2x_i - 3x_i^2 \ge 0$ for $x_i\le \frac23$, we obtain our first characterization of the biased student's behavior, in contrast to the rational student's equally-spaced portfolio:

\begin{observation} \label{obs:Deltagaps}
    Above a threshold of $\frac{2}{3}$, $\Delta_{i+1} > \Delta_{i}$, so the gaps in competitiveness between successive schools increase in size. Below the threshold of $\frac23$, $\Delta_{i+1} < \Delta_{i}$, so the gaps in competitiveness are decreasing in size. (See Figure \ref{fig:deltas} for an example)
\end{observation}

\begin{figure}[h!]
    \centering
    \includegraphics[width=0.8\linewidth]{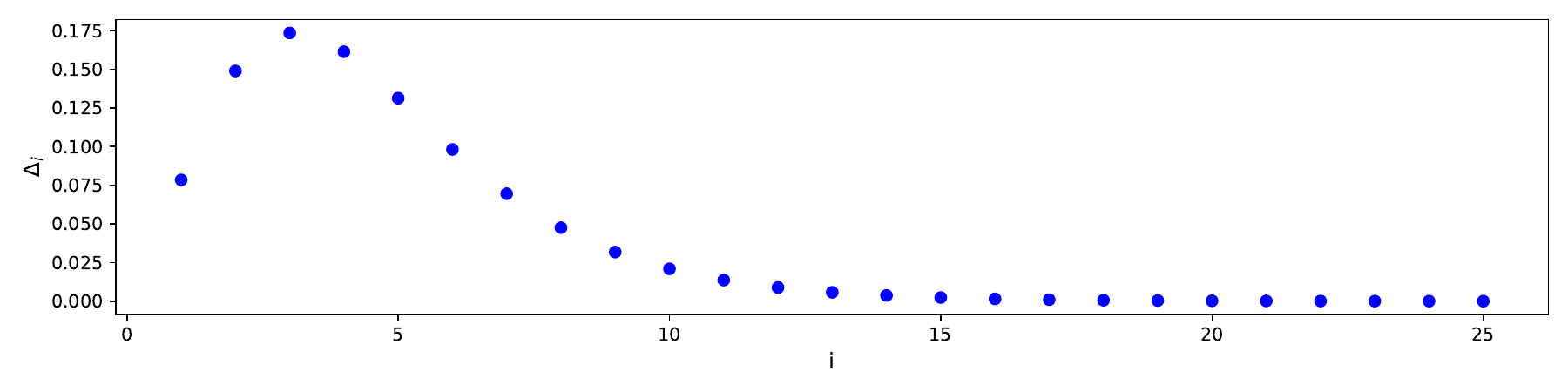}
    \caption{A plot of $\Delta_i$ for $\leq i \leq 25$ in an optimal biased portfolio for $k=25$ and $\gamma=0.1$.}
    \label{fig:deltas}
\end{figure}

Next, we observe that the perceived expected utility strictly increases with $k$. Thus, without application costs, a $\gamma$-biased student would apply to infinitely many schools. 

\begin{proposition} \label{prop:monotone_k}
    For any fixed $\gamma$, the perceived expected utility is strictly increasing in $k$. 
\end{proposition}
% \etcomment{I proposed below a compressed version of this proof. Let me know what you think} \ercomment{seems good}
\begin{proof} 
Let $OPT(k)$ denote the perceived expected payoff of the optimal $k$-portfolio. To prove the proposition, we need to argue that there exists a $k+1$-portfolio with higher payoff than $OPT(k)$. To see this, consider the lowest school $x_k$ in the optimal $k$-portfolio. Consider the \soedit{marginal} utility \soedit{from} applying to an additional school $y<x_k$: the negative bias term is $\gamma y^2(1-y)$ while the gain in expected consumption utility is $y(x_k-y)$. For a small enough $y$,
the gain is linear in $y$ while the loss is quadratic in $y$. Hence, there exists an extremely low ranked school $y$ that applying to increases the expected utility.
%a small school $y$ that applying to will increase the student's utility. for which this is always positive, showing that applying to an additional extremely low ranked school $y$ does increase the expected utility. 
\end{proof}
% \etdelete{
% \begin{proof}
%     Let $OPT(k)$ denote the perceived expected payoff of the optimal $k$-portfolio. Call the schools in the optimal $k-1$-portfolio $x_1 > x_2 > \dots x_{k-1}$, so that
%     $$ OPT(k-1) = \sum_{i=1}^{k-1} \left[ -\gamma (1 - x_i)x_i^2 + x_i (x_{i-1} - x_i) \right]. $$

%     For any $y \in (0, x_{k-1})$, $\{x_1, x_2, \dots, x_{k-1}, y\}$ is a possible $k$-portfolio of perceived expected value 
%     \begin{align*}
%         &\sum_{i=1}^{k-1} -\gamma (1 - x_i)x_i^2 - \gamma (1 - y) y^2 + \sum_{i=1}^{k-1} x_i (x_{i-1} - x_i) + y(x_{k-1} - y) \\
%         &= OPT(k-1) - \gamma (1 - y) y^2 + y(x_{k-1} - y) \\
%         &= OPT(k-1) + y(-\gamma y(1-y) + x_{k-1} - y).
%     \end{align*}
%     Observe that the inner quadratic in $y$ is equal to $x_{k-1} > 0$ at $y = 0$. By continuity there exists $y'$ sufficiently close to (but strictly greater than) $0$ such that $-\gamma y'(1-y') + x_{k-1} - y' > 0$, and consequently $y'(-\gamma y'(1-y') + x_{k-1} - y') > 0$. Then, the optimal $k$-portfolio is at least as good as $\{x_1, \dots, x_{k-1}, y'\}$, so
%     \begin{align*}
%         OPT(k) &\ge OPT(k-1) + y'(-\gamma y'(1-y') + x_{k-1} - y') \\
%         &> OPT(k-1).
%     \end{align*}    
% \end{proof}
% }

\section{Limited applications to top schools and its consequences}
%\section{Avoiding selective schools and its consequences} 
\label{sec:undershooting}

\subsection{Number of applications to top schools}

% \ercomment{this subsection heading (and the first sentence) feels a bit confusing, because first we show an ``absolute upper bound'' and then a ``too low number''. maybe it should just be ``Applications to top schools'', or ``Limited applications to top schools'', or ``Under-applying to / undershooting top schools''?}\etcomment{I thought the subsection heading is fine but this is a problem with the whole section title, they are not "avoiding" top schools. I proposed a changed title.}

%\ercomment{first we talk about the absolute upper bound, then the number of schools later -- suggesting to move the first half of this sentence to later in the subsection, or remove entirely (since we already have some text before Prop 5.3)}
One way to quantify a precise sense in which a $\gamma$-biased student ``under-applies'' is to show that there is an absolute upper bound on the selectiveness of (most of) their applications. That is, even when she is allowed an infinite number of applications, the second highest school in her portfolio remains strictly below $1$ (in contrast, the second highest rational school is $\frac{k-1}{k+1}$, which converges to $1$ as $k$ grows large). 

\begin{proposition}  \label{prop-x2-bounded}
   For any portfolio size $k$, $x_2 \leq h(\gamma)$:
   \begin{equation}
h(\gamma) =
\begin{cases}
1 - \gamma &  \gamma \in \left[0, \frac{1}{2}\right) \\
    \frac{1-\gamma+\gamma^2}{3\gamma} & \gamma \in \left[\frac{1}{2}, 2\right) \\
    \frac{1+2\gamma+ \sqrt{(1-2\gamma)^2-4\gamma}}{6\gamma} & \gamma \geq 2 \\
\end{cases}
\end{equation}
The function $h(\gamma)$ is illustrated in Figure \ref{fig:x2-bound}. 
\end{proposition}
\begin{proof}
By rearranging Equation (\ref{eq-der-0}) we get that the following holds for $x_2$:
\begin{align*}
    x_2 = 2x_1 - 1 + \gamma(2x_1 - 3x_1^2) &= (-3\gamma) x_1^2 + (2+2\gamma) x_1 - 1 .
\end{align*}
We need to find the maximum of this over $x_1 \in [0,1]$ under the constraint that $x_2 \leq x_1$. By taking a derivative we get that the unconstrained maximum of the function $(-3\gamma) x_1^2 + (2+2\gamma) x_1 - 1 $ is $\frac{\gamma}{3} + \frac{1}{3\gamma} - \frac13,$ which is achieved at $x_1 = \frac{1+\gamma}{3\gamma}$. For $\gamma \le \frac12$, $\frac{1+\gamma}{3\gamma} \ge 1$, so the maximum over $x_1 \in [0,1]$ occurs at $x_1 = 1$, which gives the bound of $x_2 \le 1-\gamma.$ 

For $1/2<\gamma < 2$ we have $\frac{1+\gamma}{3\gamma} < 1$, and additionally that: $x_2 = \frac{1-\gamma+\gamma^2}{3\gamma} \leq \frac{1+\gamma}{3\gamma}$,
and hence the unconstrained maximal value of $x_2$ is indeed achieved respecting the constraint $x_1 \in [x_2,1]$, giving the bound $x_2\le \frac{1-\gamma+\gamma^2}{3\gamma}$. 

For $\gamma>2$, the maximum computed above is invalid as it gives $x_1<x_2$. Instead, we look for the maximal value of $(-3\gamma) x_1^2 + (2+2\gamma) x_1 - 1$ over the interval in which $(-3\gamma) x_1^2 + (2+2\gamma) x_1 - 1<x_1$. This holds for $x_1 \in \left[0, \frac{1+2\gamma- \sqrt{(1-2\gamma)^2-4\gamma }}{6\gamma}\right] \cup \left[\frac{1+2\gamma+ \sqrt{(1-2\gamma)^2-4\gamma }}{6\gamma}, 1\right]$, with the maximum occurring at $x_1=\frac{1+2\gamma+ \sqrt{(1-2\gamma)^2-4\gamma }}{6\gamma}$. Here, $x_2=x_1$, and we obtain the bound of $x_2\le \frac{1+2\gamma+ \sqrt{(1-2\gamma)^2-4\gamma }}{6\gamma}$.
\end{proof}
Recall that for an unbiased agent with a portfolio of $k$ schools, we have that $x_2=1-\frac{2}{k+1}$ and in general $x_i = 1-\frac{i}{k+1}$. This implies that as $k$ increases the unbiased agent would apply to more schools in each interval (a sort of ``uniform densification''). This is not the case for the biased agent. The bound of $x_2 \leq h(\gamma)$ implies the following corollary: 
\begin{corollary}
       For any $\gamma>0$ and $k>0$, the agent will not apply to more than one school above $h(\gamma)$. 
\end{corollary}

\begin{figure}[thb]
    \centering
    \includegraphics[width=0.5\linewidth]{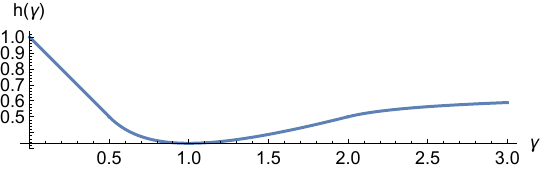}
    \caption{A plot of the function $h(\gamma)$ such that for any $k$, $x_2\leq h(\gamma)$. } 
    \label{fig:x2-bound}
\end{figure}

That is, the way in which reputation-based bias and the fear of rejection hurt a student is \emph{not} by simply causing her to stop applying altogether. Instead, a more subtle effect is at play: even given an infinite number of applications, a biased student fails to use them effectively; she applies to too few competitive schools and too many weak schools. 

\begin{proposition} \label{prop:above23bound}
 In a $k$ schools portfolio, the number of schools that the student applies to {above $2/3$} is at most $1$ for $\gamma>\frac{1}{3}$, and at most  $1 + \frac{1}{3\gamma}$ for $0< \gamma < \frac{1}{3}$.
\end{proposition} 
\begin{proof}
Suppose $x_1 \ge \frac23$ (otherwise we are trivially done). From Equation~(\ref{eq:Deltagap}) we have 
    $$\Delta_2 = \Delta_1  -\gamma (2x - 3x^2) = 1-x_1-\gamma (2x - 3x^2).$$
    Further, since by Observation~\ref{obs:Deltagaps} the $\Delta_i$'s are increasing while $x_i \ge \frac23$, we can bound the number of such $x_i$ by the following
    \begin{align*}
        c(\gamma) \le 1 + \frac{x_1 - \frac23}{\Delta_2} = 1+ \frac{x_1 - \frac23}{1-x_1-\gamma (2x - 3x^2)}
    \end{align*}
  The derivative with respect to $x_1$ of the RHS is $\frac{1 - \gamma(2-3x_1)^2}{3(3\gamma x_1^2 - (2\gamma+1)x_1+1)^2}$. For $0< \gamma <1$, the derivative is always positive on $[\frac23, 1]$, so we can take $x_1 = 1$ to get $c(\gamma) \le 1 + \frac{1}{3\gamma}.$  
  Moreover, for any $\frac{1}{3}<\gamma <1$, we have that $\frac{1}{3\gamma}<1$ and hence the student applies to at most a single school above $2/3$.
  For $\gamma>1$, the maximum on the right handside is obtained when the derivative is $0$: $1 - \gamma(2-3x_1)^2 = 0$. By solving for $x_1$ we get:
\begin{align*}
x_1=\frac{12\gamma \pm \sqrt{144\gamma^2+36\gamma(1-4\gamma)}}{18\gamma} = \frac{12\gamma \pm \sqrt{36\gamma}}{18\gamma} = \frac{2\gamma \pm \sqrt{\gamma}}{3\gamma}
\end{align*}
Thus, $x_1 = \frac{2}{3} \pm \frac{1}{3\sqrt{\gamma}}$. Note that $\frac{2}{3} - \frac{1}{3\sqrt{\gamma}} < 0$ for $\gamma>1$. The only root in the range $[2/3,1]$ is the positive one, so $x_1 = \frac{2}{3} + \frac{1}{3\sqrt{\gamma}}$. Note that it is indeed the case that the second derivative is negative for this value. Hence, we conclude that the maximum is obtained at $\frac{2}{3} + \frac{1}{3\sqrt{\gamma}}$ and:
\begin{align*}
c(\gamma) \leq 1 + \frac{\frac{1}{3\sqrt{\gamma}}}{1 - \frac{2}{3} - \frac{1}{3\sqrt{\gamma}} - \gamma(2(\frac{2}{3} + \frac{1}{3\sqrt{\gamma}}) - 3(\frac{2}{3} + \frac{1}{3\sqrt{\gamma}})^2)} = \frac{2 \gamma+2 \sqrt{\gamma}}{2 \gamma+2 \sqrt{\gamma}-1}
\end{align*}
Finally, we note that for $\gamma{>}1,$ $\frac{2 \gamma+2 \sqrt{\gamma}}{2 \gamma+2 \sqrt{\gamma}-1} < 2$, so the agent applies to at most one school above $2/3$. %\etcomment{something seems wrong here. is this true? maybe you mean $\gamma\ge 1$, i switched it}
\end{proof}

\subsection{Application to bounded number of schools above every constant}

In fact, for \emph{any} constant $c$, the number of schools that a $\gamma$-biased student applies above $c$ is bounded by a constant. To get this result we compare the potential benefit from applying to a school with the fixed subjective cost of the application and conclude that the number of schools that the student applies to in $[c, \frac23]$ is bounded by a constant.
\begin{proposition} \label{prop-const-below-23}
For any $k>0$ the student applies to at most $\frac{2/3-c}{\sqrt{\gamma c^2(1-c)}}+1$ schools in $[c,2/3]$.
\end{proposition}
\begin{proof}
%In fact, we prove a stronger claim showing that for any value of $\gamma>0$ and any $c>0$, the number of schools that a student that applies to in the interval $[c,2/3]$ is non-increasing \etreplace{in}{as a function of her overall application} $k$\etedit{, assuming $k>\frac{1}{\sqrt{\gamma c^2(1-c)}}$}. Formally, 
Let $n(c,k)$ denote the number of schools that an agent, who is applying to $k$ schools applies to in the interval $[c,2/3]$. We would like to show a limit on $n(c,k)$. Consider a portfolio of $k$ schools $x_1, \ldots, x_k$ with $n(c,k)$ schools in the interval $[c,2/3]$. Assuming $n(c,k)\ge 3$, there must be a school $x_i\in [c,2/3]$ such that 
$$x_{i+1}-x_{i-1}= (x_{i+1} - x_i)+(x_i-x_{i-1}) \le \frac{2(2/3-c)}{n(c,k)-1}.$$ 
Consider the marginal contribution loss when dropping this school $x_i$ from the portfolio. The lost consumption utility is 
    \begin{align*}
    (x_{i-1}-x_i)x_i -  (x_{i-1}-x_i)(x_{i+1}) = (x_{i-1}-x_i)(x_i-x_{i+1}) \le \frac{(2/3-c)^2}{(n(c,k)-1)^2}.
    \end{align*}
    %Since $x_{i+1}-x_{i-1} \le \frac{2(2/3-c)}{n(c,k)-1}$,
%By the bound on $(x_{i+1}-x_{i-1})$ this implies that 
%the decrease in consumption utility is at most $\frac{(2/3-c)^2}{(n(c,k)-1)^2}$. 

Now consider the subjective cost  $\gamma(x^2(1-x))$ of applying to a school $x$. Notice that $\gamma(x^2(1-x))$ is increasing in $[c,2/3]$. Thus, the subjective cost for applying to a school in $[c,2/3]$ is at least $\gamma(c^2(1-c))$. This implies that if
$\gamma(c^2(1-c)) > \frac{(2/3-c)^2}{(n(c,k)-1)^2}$
then dropping school $x_i$ improves the portfolio. This is impossible according to Proposition \ref{prop:monotone_k}, %which cannot happen in the optimal solution by Proposition \ref{prop:monotone_k}
implying  
$n(c,k)\le\frac{{2/3-c}}{\sqrt{\gamma c^2(1-c)}}+1.$
\end{proof}
% \etdelete{To simplify the inequality a bit we instead solve for
% \begin{align*}
%     \gamma(c^2(1-c)) > \frac{1}{k^2} \implies k > \frac{1}{\sqrt{\gamma c^2(1-c)}}
% \end{align*}
% Thus, for any $k>\frac{1}{\sqrt{\gamma c^2(1-c)}}$ the number of schools that the agent applies to in $[c,2/3]$ does not increase (e.g., $n(c,k+1) \leq n(c,k)$). This also trivially implies that $n(c,k) \leq \frac{1}{\sqrt{\gamma c^2(1-c)}}$.}

Putting this together with Proposition \ref{prop:above23bound} that provides a bound on the number of schools that the student applies to above $2/3$, we conclude that:
\begin{restatable}{corollary}{ngammac}\label{cor:n_gamma_c}
     A $\gamma$-biased student applies to at most ${m(\gamma, c)} = 2+ \frac{1}{3\gamma} + \frac{{2/3-c}}{\sqrt{\gamma c^2(1-c)}}$ schools above $c$ for $0<\gamma<\frac{1}{3}$, and at most ${m(\gamma, c)} = {2}+\frac{{2/3-c}}{\sqrt{\gamma c^2(1-c)}}$ schools for $\gamma\geq \frac{1}{3}$.
\end{restatable}

To get a better sense of this bound, in Figure \ref{fig:const-schools} we plot this bound on the number of schools above $c$ as a function of $c$ for several values of $\gamma$.
\begin{figure}[htp]
    \centering
    \includegraphics[width=0.5\linewidth]{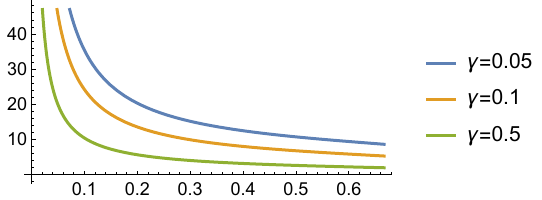}
    \caption{A bound on the number of schools above $c$ that a student applies to as a function of $c$.}
    \label{fig:const-schools}
\end{figure}

\subsection{Limited utility of biased students}
%\socomment{I changed it to theorem, like we said}
\begin{theorem}
    For every $\gamma > 0$, the true expected payoff of a student with bias $\gamma$ is bounded above by a constant $p(\gamma) < \frac{1}{2}$. Consequently, for every $\gamma > 0$, there exists $k = k(\gamma)$ such that a rational student limited to $k$ applications achieves strictly higher expected payoff than a student with bias $\gamma$ with an unlimited number of applications.
\end{theorem}
\begin{proof}
    First, recall from Corollary~\ref{cor:unbiased1/2} that as $k$ increases, the expected payoff of an unbiased agent converges to $\frac12$. % \etdelete{An alternative proof of this fact follows from considering the continuous analogue of an infinitely large, equally-spaced portfolio, and integrating with uniform density over $[0,1]$ to obtain $\EE [U_{OPT}] = \int_{0}^1 x dx = \frac{1}{2}$.}\etcomment{this is really the essence of the proof of Corollary~\ref{cor:unbiased1/2}, so no benefit from repeating here.} 
    This can also be seen by considering the continuous analogue of an infinitely large equally spaced portfolio, and integrating with uniform density over $[0,1]$. By conditioning on the student's score $x$ and splitting the integral at any constant $c$, we observe: 
    \begin{align*}
            \EE [U_{OPT}] &= \EE \left[U_{OPT} \mid x \le c \right] \cdot \PP \left( x \le c \right) +  \EE \left[U_{OPT} \mid x > c \right] \cdot \PP \left( x > c \right) \\
            &=  \EE \left[U_{OPT} \cdot \1_{\{x \le c \}} \right] + \EE \left[U_{OPT} \cdot \1_{\{x > c \}} \right] \\
            &= \int_0^{c} x dx + \int_{c}^{1} x dx = \frac{c^2}{2} + \frac{1-c^2}{2} = \frac12.
        \end{align*}

    Now, consider a student with bias $\gamma > 0$ applying to an infinitely large portfolio of schools ($k\to \infty$), by Corollary \ref{cor:n_gamma_c} {all but a finite} number of which is below $c$. If the student's score is below $c$, her true expected payoff can be no better than the unbiased optimal utility from a uniform continuum on $c$ (and in fact may be worse, since her bias causes her to apply suboptimally). That is, 
        $$\EE\left[U \cdot \1_{\{x \le c\}} \right] \le \EE \left[U_{OPT} \cdot \1_{\{x \le c \}} \right] = \frac{c^2}{2}.$$
        
    In the case where the student's score is above $c$, she only has a finite, constant number of chances $m=m(\gamma, c)$ to apply to and attend a school better than $\soedit{c}$. If she were to apply optimally, by extending the equal spacing part of Proposition \ref{prop:spacing} to intervals, we can conclude that she would choose the equally spaced set $c + \frac{i(1-c)}{m+1}$ for $i\in [m]$. However, the student is not necessarily applying optimally, so this set of schools provides an upper bound on her expected payoff. That is, 
        \begin{align*}
             \EE \left[U \cdot \1_{\{x > c \}} \right] \le \sum_{i=0}^{m} \left(c + \frac{i(1-c)}{m+1} \right) \frac{1-c}{m+1} &= \frac{c(1-c)}{m+1}(m+1) + \frac{(1-c)^2}{(m+1)^2} \frac{m(m+1)}{2} \\
             &= c(1-c) + \frac{n(1-c)^2}{2(m+1)} \\
             &= \frac{1-c}{2} \left[2c + 1-c 
 - \frac{1-c}{m+1} \right] \\
             &= \frac{1-c^2}{2} - \frac{(1-c)^2}{2(m+1)}.
        \end{align*}

    In total, 
    \begin{align*}
            \EE \left[U \right] &= \EE \left[U \cdot \1_{\{x \le c\}} \right] + \EE \left[U \cdot \1_{\{x > c\}} \right] \\
            &\le \frac{c^2}{2} + \frac{1-c^2}{2} - \frac{(1-c)^2}{2(m(\gamma, c)+1)} =\frac12- \frac{(1-c)^2}{2(m(\gamma, c)+1)} \eqqcolon p(\gamma) \\
            &< \frac12 = \EE \left[U_{OPT}\right],
        \end{align*}
        where the inequality is strict because $m(\gamma, c)$ is finite (by Corollary~\ref{cor:n_gamma_c}). 
    \begin{comment}
    \begin{figure}
        \centering
        \includegraphics[width=0.5\linewidth]{utility-fixed-gamma.pdf}
        \caption{\soedit{This figure is the bounded we get on the utility as a function of $c$ it tells us that for every $\gamma$ there is an optimal value of $c$ which we should use in order to get the best bound of this argument. I added this plot first to make sure that everything make sense, you can see what you think and second the question is whether it is useful to include this here.}}
        \label{fig:enter-label}
    \end{figure}

    \begin{figure}[H]
    \centering
    \includegraphics[width=0.5\linewidth]{utility-fixed-c.pdf}
    \caption{\soedit{Another one showing for different values of $\gamma$ which $c$ works better}}
    \label{fig:enter-label}
\end{figure}
\end{comment}

    Finally, observe that a rational student with $k$ applications obtains expected payoff $\frac{k}{2(k+1)}$. For $k > k(\gamma) = \frac{2p(\gamma)}{1 - 2p(\gamma)}$, this value is larger than $p(\gamma)$, meaning that even with an unlimited number of applications, the biased student has a lower expected payoff compared to a rational student with a portfolio size constraint. 
    
    %}

\end{proof}

\subsection{Computational results} \label{sec:comp_obs}

Finally, we highlight further interesting computational results that complement our theoretical analysis of the biased student's strategy across varying values of $k$ and $\gamma$ (detailed numerical results are provided in Appendix~\ref{sec:tables}). In line with Propositions~\ref{prop:above23bound} and~\ref{prop-const-below-23}, we observe \emph{downwards expansion}: given an increasing number of applications, the biased student only sends more applications to schools increasingly close to $0$.

\begin{observation} (computational)
    For fixed $\gamma > 0$, as $k$ increases, the values $x_i$ of the highest schools in the portfolio (and consequently, the gaps $\Delta_i$ between them) approximately ``freeze'' in place. That is, a student with bias $\gamma > 0$ only expands her portfolio \emph{downwards}, by sending out an increasing number of ``trivial'' applications close to $0$.
\end{observation}

This downwards expansion phenomenon stands in sharp contrast with the uniform densification of the rational student's portfolio (analogous to the reach/match/safety strategy described by \cite{as2023college}). As a result, increasing the number of applications $k$ is less helpful to the payoff of a $\gamma$-biased student than of a rational student; this effect becomes more pronounced as $\gamma$ increases.

\begin{observation} (computational)
    A student with bias $\gamma > 0$ quickly obtains most of her true expected utility within a small number of applications.
\end{observation}

\section{Approximately solving for the $\gamma$-biased student's portfolio} \label{sec:gamma-0-convergence} \label{sec:approx}

In this section, we establish upper and lower bounds on each school in the portfolio of a $\gamma$-biased student. Then, as we take $\gamma$ to $0$, we see that the biased portfolio smoothly converges to the optimal rational portfolio. 
These results will also be useful in analyzing a student's overshooting behavior in the next section.
As a first step, {in Appendix \ref{app-approx} we use induction to} bound each school $x_i$ relative to its neighbor $x_{i-1}$:
%\subsection{Approximately Solving for the Optimal Portfolio for a $\gamma$-biased Agent}

\begin{lemma} \label{lem-xi-bounds}
For every $1 \leq i \leq k$, let $\bar \delta_i$ and $\underline \delta_i$ be such that: $\underline{\delta_i} \leq -(2x_i-3x_i^2) \leq \bar \delta_i$:
\begin{align*}
\frac{k+1-i}{k+2-i} x_{i-1} + \frac{\gamma}{k+2-i} \sum_{j=1}^{k+1-i} j \cdot  \underline{\delta}_{k+1-j} \leq x_{i} \leq \frac{k+1-i}{k+2-i}  \cdot x_{i-1} + \frac{\gamma}{k+2-i} \sum_{j=1}^{k+1-i} j \cdot  \bar{\delta}_{k+1-j}.
\end{align*}
\end{lemma}

Using Lemma~\ref{lem-xi-bounds}, we now establish the following bounds on each $x_i$ solely as a general function of $i$, $k$, and $\gamma$:

\begin{theorem} \label{thm-bounds-general}
For every $x_i$, let $\bar \delta_i$ and $\underline \delta_i$ be such that: $\underline{\delta_i} \leq -(2x_i-3x_i^2) \leq \bar \delta_i$. Then, the following holds for any $1 \leq i \leq k$:
\begin{align*}
 x_i &\geq \frac{k+1-i}{k+1} +\gamma \cdot \frac{k+1-i}{k+1} \sum_{j=1}^{i-1} j \cdot \underline{\delta}_j + \gamma \cdot \frac{i}{k+1} \sum_{j=1}^{k+1-i} j \cdot \underline{\delta}_{k+1-j},    \\
 x_i &\leq \frac{k+1-i}{k+1} + \gamma \cdot \frac{k+1-i}{k+1} \sum_{j=1}^{i-1} j \cdot \bar{\delta}_j + \gamma \cdot \frac{i}{k+1} \sum_{j=1}^{k+1-i} j \cdot \bar \delta_{k+1-j}.
\end{align*}
\end{theorem}
\begin{proof}
We use  Lemma \ref{lem-xi-bounds} to solve for $x_i$ and prove by induction that the theorem holds. As the proof is the same for the upper and lower bounds, we will only write the proof for the upper bound.
The base case is $i=1$. We have by Lemma \ref{lem-xi-bounds}: 
\begin{align*}
x_1 &\leq\frac{k}{k+1} + \gamma \cdot \frac{1}{k+1} \sum_{j=1}^{k} j \cdot \bar \delta_{k+1-j}. 
\end{align*}
Next, we assume the induction hypothesis holds for $i-1$ and prove for $i$. By Lemma \ref{lem-xi-bounds}:
\begin{align*}
     x_{i} \leq \frac{k+1-i}{k+2-i}  \cdot x_{i-1} + \frac{\gamma}{k+2-i} \sum_{j=1}^{k+1-i} j \cdot  \bar{\delta}_{k+1-j}.
\end{align*}
By using the induction hypothesis for $x_{i-1}$ we get:
\begin{align*}
    x_{i} &\leq \frac{k+1-i}{k+2-i}  \left(\frac{k+2-i}{k+1} + \gamma \cdot \frac{k+2-i}{k+1} \sum_{j=1}^{i-2} j \cdot \bar{\delta}_j + \gamma \cdot \frac{i-1}{k+1} \sum_{j=1}^{k+2-i} j \cdot \bar \delta_{k+1-j} \right) + \frac{\gamma}{k+2-i} \sum_{j=1}^{k+1-i} j \cdot  \bar{\delta}_{k+1-j} \\
    &\leq \frac{k+1-i}{k+1} + \gamma \frac{k+1-i}{k+1} \sum_{j=1}^{i-2} j \cdot \bar{\delta}_j + \gamma \frac{k+1-i}{k+2-i} \cdot \frac{i-1}{k+1} \sum_{j=1}^{k+2-i} j \cdot \bar \delta_{k+1-j} + \frac{\gamma}{k+2-i} \sum_{j=1}^{k+1-i} j \cdot  \bar{\delta}_{k+1-j}.
\end{align*}
Note that $\bar{\delta}_{i-1}$ only appears in the second sum with a coefficient of: 
\begin{align*}
\gamma \frac{k+1-i}{k+2-i} \cdot \frac{i-1}{k+1} \cdot (k+2-i) = \gamma \frac{k+1-i}{k+1} \cdot (i-1).
\end{align*}
Also, for every $j\geq i$ the coefficient of $\bar{\delta}_j$ is (from the second and third summations):
\begin{align*}
    \gamma \cdot j \cdot \left( \frac{k+1-i}{k+2-i} \cdot \frac{i-1}{k+1} + \frac{1}{k+2-i} \right) = \gamma \cdot j \cdot \frac{i}{k+1}.
\end{align*}
Thus, we have that:
\begin{align*}
x_i &\leq \frac{k+1-i}{k+1} + \gamma \cdot \frac{k+1-i}{k+1} \sum_{j=1}^{i-1} j \cdot \bar{\delta}_j + \gamma \cdot \frac{i}{k+1} \sum_{j=1}^{k+1-i} j \cdot \bar \delta_{k+1-j}.
\end{align*}
\end{proof}

Now, as a first application of Theorem \ref{thm-bounds-general} we observe that for every $x_i$ we have that the required bounds hold with $\bar{\delta}_i=1$ and $\underline{\delta_i}=1/3$, which establishes the following bounds on each $x_i$:
\begin{theorem} \label{thm-bounds}
For every $k$ and every $\gamma>0$ we have that, for every $1\leq i \leq k$:
\begin{align*}
 \frac{k+1-i}{k+1} - \frac{\gamma}{3} \cdot \frac{i(k+1-i)}{2} \leq   x_i \leq \frac{k+1-i}{k+1} + \gamma \cdot \frac{i(k+1-i)}{2}.
\end{align*}
\end{theorem}
\begin{proof}
Note that for any $x_i$, we have that  $-1/3 \leq -(2x_i - 3x_i^2) \leq 1$. Thus, we can apply Theorem \ref{thm-bounds-general} with $\bar{\delta}_i=1$ and $\underline{\delta_i}=-1/3$ and get that:
\begin{align*}
    x_i &\leq \frac{k+1-i}{k+1} + \gamma \cdot \frac{k+1-i}{k+1} \sum_{j=1}^{i-1} j + \gamma \cdot \frac{i}{k+1} \sum_{j=1}^{k+1-i} j \\
    &= \frac{k+1-i}{k+1} + \gamma \cdot \frac{k+1-i}{k+1} \cdot\frac{(i-1)(i)}{2} + \gamma \cdot \frac{(2+k-i)(k+1-i)}{2} \\
    &= \frac{k+1-i}{k+1} + \gamma \frac{k+1-i}{k+1} \left( \frac{(i-1)(i)}{2} + \frac{i(2+k-i)}{2} \right) \\
     &= \frac{k+1-i}{k+1} + \gamma \frac{k+1-i}{k+1} \cdot \frac{i(k+1)}{2} \\
     &= \frac{k+1-i}{k+1} + \gamma \frac{i(k+1-i)}{2}. 
\end{align*} 
The proof for the lower bound follows identically.
\end{proof}
By observing that $\frac{i(k+1-i)}{2} \leq \frac{(k+1)^2}{8}$, we obtain this crude bound on the $x_i$'s: 
\begin{align*}
 \frac{k+1-i}{k+1} -\frac{(k+1)^2}{8} \cdot \frac{\gamma}{3} \leq   x_i \leq \frac{k+1-i}{k+1} + \frac{(k+1)^2}{8} \gamma.
\end{align*}
These bounds imply the following convergence:
\begin{corollary} \label{cor-conv-fixed-k}
For a fixed $k$, as $\gamma$ goes to $0$ then $x_i$ converges to $\frac{k+1-i}{k+1}$.
\end{corollary}
Interestingly, as we will see in Section \ref{sec:mon}, this convergence is not monotone for values of $i$ such that $\frac{k+1-i}{k+1} > \frac{2}{3}$.

\section{Fear of rejection can also cause overshooting} \label{sec:mon}

% \etcomment{Previous title: A surprising phenomenon: Non-monotonicity with respect to $\gamma$} \ercomment{thoughts on removing the parenthetical (choices not monotone in $\gamma$) (so that the title fits on one line)?} \etcomment{great idea}

Intuitively, one might expect that as $\gamma$ increases, the biased student becomes increasingly sensitive to the fear of rejection and undershoots the optimal portfolio more and more. However, we find that this is \emph{not} always the case. In this section, we show that there exist instances where, although the biased portfolio does indeed converge to the optimal portfolio as $\gamma$ shrinks to $0$, this convergence is not monotone -- decreasing $\gamma$ first causes an increase in competitiveness {of some choices}, then a decrease. Perhaps even more surprisingly, in some cases the biased agent actually applies to a school that is \emph{higher} than optimal -- indicating that a biased student may underperform not only by undershooting the optimal portfolio, but also sometimes by \emph{overshooting}. {Such a student may be reacting to fear of rejection, by applying to schools where she is very unlikely to be accepted, and hence with small expectation, limiting her disappointment when rejected.}

\subsection{Local overshooting}

The following proposition illustrates some of the required conditions for a biased student to apply higher than an unbiased student.

\begin{proposition} \label{prop-local-over}
Let $\gamma<2-\sqrt{3}$. Consider a biased agent that applies to schools $a$ and $b$ and doesn't apply to any school in between. Let $x$ be the optimal school for this biased agent to apply to in $(a,b)$. Then, $x>\frac{a+b}{2}$ if $a+b>4/3$ and $x < \frac{a+b}{2}$ if $a+b\leq 4/3${, while an unbiased agent applies to $x=\frac{a+b}{2}$}.
\end{proposition}
% \etcomment{you appear to assume in the proof that the optimal school $x$ exists. If $a+b$ is high, they may not. it also comes out in the proof, as $(1+\gamma)^2-3\gamma z$ may be negative.  } \socomment{I guess I meant that $a$ and $b$ are also schools that the student applies to as part as an optimal portfolio and because of this reason it should be the case that $(1+\gamma)^2-3\gamma z$ is not negative. Another option is just to have this proposition for $\gamma< 2-\sqrt 3$.  }\etcomment{I added inside the proof that $z\le 2-\gamma$ by Proposition \ref{prop-x2-bounded}. you took it out? why?}
\begin{proof}
Let $z=a+b$. By Equation (\ref{eq-der-0}) we have that:
$2x-z+\gamma(2x-3x^2)=0$. If $\gamma=0$ then we get $x=\frac{z}{2}=\frac{a+b}{2}$, extending the equal spacing of  Proposition \ref{prop:spacing}. If $\gamma>0$ this is a quadratic equation for $x$, and $x=\frac{z}{2}$ is not a solution. {For $\gamma<2-\sqrt{3}$}
%\etedit{For $\gamma<1/2$, we know from Proposition \ref{prop-x2-bounded} that $z<1-\gamma$,} so 
we have that:
\begin{align*}
    x= \frac{1+\gamma - \sqrt{(1+\gamma)^2-3\gamma  z}}{3\gamma}
\end{align*}
We would like to find the condition under which $x\leq \frac{z}{2}$
\begin{align*}
    x= \frac{1+\gamma - \sqrt{(1+\gamma)^2-3\gamma  z}}{3\gamma} \leq \frac{z}{2} &\implies 1+\gamma - \sqrt{(1+\gamma)^2-3\gamma  z} \leq \frac{3\gamma z}{2} \\
    &\implies 1+\gamma -\frac{3\gamma z}{2} \leq  \sqrt{(1+\gamma)^2-3\gamma  z} 
\end{align*}
% By rearranging we get:
% \begin{align*}
% 1+\gamma - \sqrt{(1+\gamma)^2-3\gamma  z} \leq \frac{3\gamma z}{2}
% \end{align*}
% This implies that:
% \begin{align*}
% 1+\gamma -\frac{3\gamma z}{2} \leq  \sqrt{(1+\gamma)^2-3\gamma  z} 
% \end{align*}
Note that since $z<2$ and $\gamma<1/2$ both sides of the inequality are positive  and hence it implies that:
\begin{align*}
(1+\gamma -\frac{3\gamma z}{2})^2 \leq  (1+\gamma)^2-3\gamma  z
\end{align*}
By rearranging we get that:
\begin{align*}
    (1+\gamma)^2-2(1+\gamma)\frac{3\gamma z}{2} + \frac{9\gamma^2 z^2}{4} \leq (1+\gamma)^2-3\gamma  z 
&\implies -2(1+\gamma)\frac{3}{2} + \frac{9\gamma z}{4} \leq -3 \\
%&\implies  \frac{9\gamma z}{4} \leq 3\gamma \\
&\implies z \leq 4/3
\end{align*}
This concludes the proof. 
\end{proof}
The above proof also holds for $\gamma<1/2$ if we take into account that by Proposition \ref{prop-x2-bounded} for any two schools $a$ and $b$ that are part of an optimal portfolio we have that $z<2-\gamma$.
Furthermore, the previous proof can be extended to achieve a more general result:
% \etdelete{\begin{proposition}
% Let $\gamma<2-\sqrt{3}$. Consider a biased agent that applies to schools $a$ and $b$ and doesn't apply to any school in between. Let $x$ be the optimal school for this biased agent to apply to in $(a,b)$. Then, $$\frac{1}{2(1+\gamma)}(a+b) \leq x \leq (a+b) \cdot \frac{1+
% \gamma - \sqrt{(1+\gamma)^2-6\gamma}}{6\gamma}$$
% \end{proposition}
% }
\begin{proposition}
Let $\gamma<2-\sqrt{3}$. Consider a biased agent that applies to schools $a$ and $b$ and doesn't apply to any school in between. Let $x$ be the optimal school for this biased agent to apply to in $(a,b)$. Then for every $$\frac{1}{2(1+\gamma)} < \theta < \frac{1+
\gamma - \sqrt{(1+\gamma)^2-6\gamma}}{6\gamma}$$ 
there exists $z(\theta)$ such that if $a+b>z(\theta)$ then $x > \theta (a+b)$ and else $x < \theta (a+b)$. 
\end{proposition}

\begin{proof}
%  \etdelete{   Another way of phrasing this statement is that for every $\frac{1}{2(1+\gamma)} < \theta < \frac{1+
% \gamma - \sqrt{(1+\gamma)^2-6\gamma}}{6\gamma}$, we have that there exists $z(\theta)$ such that if $a+b>z(\theta)$ then $x > \theta (a+b)$ and else $x < \theta (a+b)$.} \etcomment{interesting statement, but this is not what the proposition says. I proposed to move this "Another way" to be the proposition.}
This bounds are illustrated in Figure \ref{fig:theta-bounds}.
\begin{figure}
    \centering
    \includegraphics[width=0.5\linewidth]{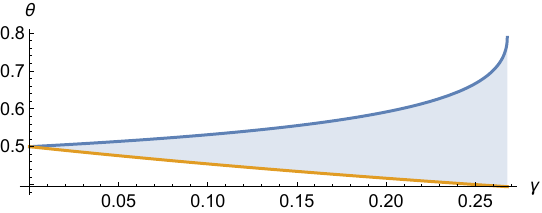}
    \caption{Illustration of $\frac{1}{2(1+\gamma)} < \theta < \frac{1+
\gamma - \sqrt{(1+\gamma)^2-6\gamma}}{6\gamma}$}
    \label{fig:theta-bounds}
\end{figure}
The proof is the same as the proof of Proposition \ref{prop-local-over} except that we solve for:
\begin{align*}
    x= \frac{1+\gamma - \sqrt{(1+\gamma)^2-3\gamma  z}}{3\gamma} \leq \theta z.
\end{align*}
By the same steps as the previous proof this gets us that $x \le \theta z$ if and only if $z$ satisfies %\begin{align*}
    $z \leq \frac{6\theta(1+\gamma)-3}{9\theta^2 \gamma} $.
%\end{align*}
To compute which values of $\theta$ are feasible, we observe $0 
\leq z \leq 2$. Since $z> 0$ we get:
\begin{align*}
    6 \theta(1+\gamma)-3>0 \implies \theta > \frac{1}{2(1+\gamma)}.
\end{align*}
Now requiring that $z < 2$ gets us that $ 6 \theta(1+\gamma)-3 < 9\theta^2 \gamma$. After rearranging we get $6 \gamma \theta^2 - 2\theta (1+\gamma) +1 >0.$
% \begin{align*}
%  6 \theta(1+\gamma)-3 < 9\theta^2 \gamma
% \end{align*}
% After rearranging we get that:
% \begin{align*}
% 6 \gamma \theta^2 - 2\theta (1+\gamma) +1 >0
% \end{align*}
we conclude that $\theta < \frac{1+
\gamma - \sqrt{(1+\gamma)^2-6\gamma}}{6\gamma}$ since $\frac{1+
\gamma + \sqrt{(1+\gamma)^2-6\gamma}}{6\gamma} >1$.
\end{proof}
 
%\socomment{The upper bound can be improved since we know that $x_2<1-\gamma$ which 
%implies that $a+b < 2-\gamma$.}

\subsection{Global overshooting}
In the following theorem, we strengthen the bounds on the $x_i$'s we proved in Theorem \ref{thm-bounds} to show that there are portfolio sizes for which there exists $\gamma$ such that the biased student simultaneously overshoots at the top school she applies to \emph{and} undershoots in the minimal school she applies to. We conjecture that this is a part of a more general phenomenon in which for large enough $k$, there exists $\gamma'$ such that for any $\gamma<\gamma'$ the student overshoots (essentially) all schools $i$ such that $\frac{k+i-1}{k+1}>2/3$ and undershoots (essentially) all schools such that $\frac{k+i-1}{k+1}<2/3$. 

\begin{figure}[ht] 
    \centering
    \begin{subfigure}[b]{0.45\textwidth}
        \includegraphics[width=\linewidth]{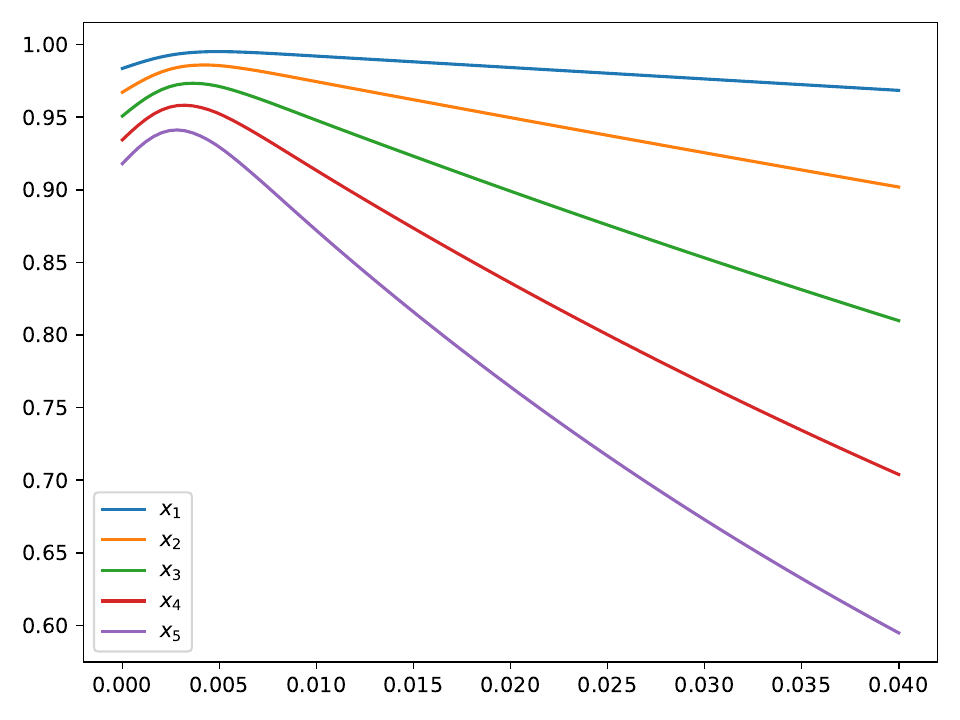}
        \caption{Top 5 schools}
        \label{fig:k-1-5-Over shooting}
    \end{subfigure}
    \hfill % This adds a horizontal space between the subfigures if needed
    \begin{subfigure}[b]{0.45\textwidth}
        \includegraphics[width=\linewidth]{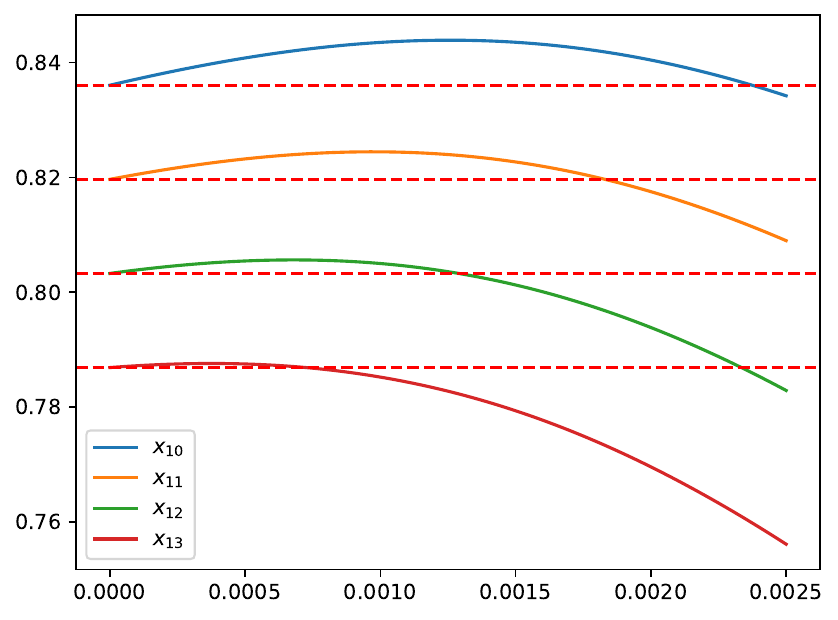}
        \caption{Schools 10-13}
        \label{fig:gamma-k-60-10-13}
    \end{subfigure}
    \caption{Overshooting top schools as a function of $\gamma$ for a portfolio of $k=60$ schools.}
    \label{fig-top-overshoot}
\end{figure}

\begin{theorem}
There exist values of $k$ for which there exists $\gamma>0$ such that for any $\gamma'<\gamma$ we have that $x_1>\frac{k}{k+1}$ and $x_k<\frac{1}{k+1}$.
\end{theorem}
\begin{proof}
Theorem \ref{thm-bounds} establishes that:
\begin{align*}
 \frac{k+1-i}{k+1} - \frac{\gamma}{3} \cdot \frac{i(k+1-i)}{2} \leq   x_i \leq \frac{k+1-i}{k+1} + \gamma \cdot \frac{i(k+1-i)}{2}. \tag{$*$}
\end{align*}
When $\gamma$ is sufficiently small we can use these bounds on the $x_i's$ to apply Theorem \ref{thm-bounds-general} with tighter bounds on $f(x_i) = -(2x_i-3x_i^2)$. To get the right bounds we should carefully consider where the function $f(x_i) = -(2x_i-3x_i^2)$ attains its minimum and maximum in an interval $[a,b]$:
\begin{enumerate}
\item If $a\geq \frac{1}{3}$, then $f(x_i)$ is increasing in the interval and hence the minimum is attained at $a$ and the maximum at $b.$
\item If $b \leq \frac{1}{3}$, then $f(x_i)$ is decreasing in the interval and hence the minimum is attained at $b$ and the maximum at $a.$
\item If $a < \frac{1}{3} < b$, then the minimum is attained at $\frac{1}{3}$ and the maximum is attained at $a$ if $a+b> \frac{2}{3}$ and at $b$ otherwise.
\end{enumerate}
To begin this proof we compute a bound on $\gamma$ such for every $x_i$ the interval defined by $(*)$ is sufficiently small and for any $x_i$ such that $\frac{k+1-i}{k+1} \neq 1/3$, the interval $x_i$ is in is either of type $1$ or of type $2$. For this reason, it is simpler to assume that $k+1$ is divisible by $3$. This implies that if we choose $\gamma$ such that $\gamma \cdot \frac{i(k+1-i)}{2} \leq \frac{1}{k+1}$ we are guaranteed that the only $x_i$ which is in an interval of type $3$ is $x_{\frac{2(k+1)}{3}}$. By observing that $\frac{i(k+1-i)}{2} > \frac{(k+1)^2}{8}$, we get that to satisfy this constraint we can pick any $\gamma \leq \frac{8}{(k+1)^3}$. 

Next, we assume that $\gamma \leq \frac{8}{(k+1)^3}$ and apply Theorem \ref{thm-bounds-general} with $\bar{\delta}_i(\gamma)$ and $\underline{\delta}_i(\gamma)$ defined as follows:
\begin{itemize}
\item For $i$ such that $\frac{k+1-i}{k+1} < 1/3$, we set $\bar{\delta}_i(\gamma)=f( \frac{k+1-i}{k+1} - \frac{\gamma}{3} \cdot \frac{i(k+1-i)}{2})$ and $\underline{\delta}_i(\gamma)= f(\frac{k+1-i}{k+1} + \gamma \cdot \frac{i(k+1-i)}{2})$
\item For $i$ such that $\frac{k+1-i}{k+1} > 1/3$, we set $\bar{\delta}_i(\gamma)= f(\frac{k+1-i}{k+1} + \gamma \cdot \frac{i(k+1-i)}{2})$ and $\underline{\delta}_i(\gamma)= f( \frac{k+1-i}{k+1} - \frac{\gamma}{3} \cdot \frac{i(k+1-i)}{2})$.
\item  For $i$ such that $\frac{k+1-i}{k+1} = 1/3$, we set $\bar{\delta}_i(\gamma)=f( \frac{k+1-i}{k+1} - \frac{\gamma}{3} \cdot \frac{i(k+1-i)}{2})$ and $\underline{\delta}_i(\gamma)=1/3$.
\end{itemize}
By applying Theorem \ref{thm-bounds-general} for $x_1$ using the $\underline{\delta}_i(\gamma)$'s we get:
\begin{align*}
 x_1  &\geq \frac{k}{k+1} + \gamma \cdot \frac{1}{k+1} \sum_{j=1}^{k} j \cdot \underline{\delta}_{k+1-j} \\
    &\geq \frac{k}{k+1} +\gamma \cdot \frac{1}{k+1} \sum_{\{j| \frac{k+1-j}{k+1}<\frac{1}{3}\}} (k+1-j) \cdot f(\frac{k+1-j}{k+1} + \gamma \cdot \frac{j(k+1-j)}{2}) \\ 
    &~~~+ \gamma \cdot \frac{1}{k+1} \frac{(k+1)}{3} \cdot f(1/3) + \gamma \cdot \frac{1}{k+1}\sum_{\{j| \frac{k+1-j}{k+1}>\frac{1}{3}\}} (k+1-j) \cdot f( \frac{k+1-j}{k+1} - \frac{\gamma}{3} \cdot \frac{j(k+1-j)}{2})
 \end{align*}
To show that there exists a threshold $\gamma'$ such that for every $\gamma < \gamma'$ we have that $x_1 > \frac{k-1}{k}$, we need to show that there exists  $\gamma'$ such that for every $\gamma < \gamma'$:
\begin{align*}
&\sum_{\{j| \frac{k+1-j}{k+1}<\frac{1}{3}\}} (k+1-j) \cdot f(\frac{k+1-j}{k+1} + \gamma \cdot \frac{j(k+1-j)}{2}) \\ 
    &~~~+ \frac{(k+1)}{3}\cdot f(1/3) + \sum_{\{j| \frac{k+1-j}{k+1}>\frac{1}{3}\}} (k+1-j) \cdot f( \frac{k+1-j}{k+1} - \frac{\gamma}{3} \cdot \frac{j(k+1-j)}{2}) >0.
\end{align*}
Consider this expression for $k=5$. We need to show that:
\begin{align*}
&f(1/6+5/2 \gamma)+2f(1/3)+3f(1/2-3/2 \gamma)+ 4f(2/3-4/3 \gamma)+5f(5/6-5/6 \gamma ) \\
&= 1/12 (5 - 362 \gamma + 849 \gamma^2) > 0 .
\end{align*}
The above expression is greater than $0$ for any $\gamma < \frac{1}{849} \left(181-2 \sqrt{7129}\right) = 0.0142912$. Recall that this argument is only valid for $\gamma \leq \frac{8}{(k+1)^3}$, hence, we get that for $\gamma \leq \min \{0.0142912, \frac{8}{6^3} \} = 0.0142912$ the top school that the biased student applies to will always be higher than $\frac{k}{k+1}$. 

In Figure \ref{fig:overshoot-k5} we plot our lower bound and $x_1$ as a function of $\gamma$. Next to it we plot $x_1$ to show that in fact the overshooting appears for substantially greater values of $\gamma$. To analytically improve our bound, we should repeatedly plug in bounds we have on the $x_i$'s to Theorem \ref{thm-bounds-general} to get increasingly finer bounds.

Note that similar techniques can also be used to obtain an upper bound on $x_k$. By Theorem \ref{thm-bounds-general}:
\begin{align*}
x_k &\leq \frac{1}{k+1} + \gamma \cdot \frac{1}{k+1} \sum_{j=1}^{k} j \cdot \bar{\delta}_j.
\end{align*}
By applying Theorem \ref{thm-bounds-general} for $x_k$ using the $\bar{\delta}_i(\gamma)$'s we get:
\begin{align*}
x_k &\leq \frac{1}{k+1} +\gamma \cdot \frac{1}{k+1} \sum_{\{j| \frac{k+1-j}{k+1}\leq \frac{1}{3}\}} j \cdot f( \frac{k+1-j}{k+1} - \frac{\gamma}{3} \cdot \frac{j(k+1-j)}{2}) \\ 
    &~~~+ \gamma \cdot \frac{1}{k+1}\sum_{\{j| \frac{k+1-j}{k+1}>\frac{1}{3}\}} j \cdot f(\frac{k+1-j}{k+1} + \gamma \cdot \frac{j(k+1-j)}{2})
\end{align*}
To show that $x_k < \frac{1}{k+1}$, we need to show that:
\begin{align*}
 \sum_{\{j| \frac{k+1-j}{k+1}\leq \frac{1}{3}\}} j \cdot f( \frac{k+1-j}{k+1} - \frac{\gamma}{3} \cdot \frac{j(k+1-j)}{2})+\sum_{\{j| \frac{k+1-j}{k+1}>\frac{1}{3}\}} j \cdot f(\frac{k+1-j}{k+1} + \gamma \cdot \frac{j(k+1-j)}{2}) < 0
\end{align*}
For $k=5$, showing that $x_5 < \frac{1}{6}$ boils down to showing that:
\begin{align*}
&5f(1/6-5/6 \gamma)+4f(1/3-4/3 \gamma)+3f(1/2+9/2 \gamma)+ 2f(2/3+4 \gamma)+f(5/6+5/2 \gamma ) \\
&= 1/12 (-35 + 494 \gamma + 3945 \gamma^2) < 0
\end{align*}
This holds for $\gamma<\frac{2 \sqrt{49771}-247}{3945} =0.0504913.$ 
{Putting this together with the bound of $\gamma < \frac{1}{6^3}$, we required at the beginning of the argument, we get that his bound holds for $\gamma< \frac{1}{6^3}$.}
\end{proof}

\begin{figure}[ht] 
    \centering
    \begin{subfigure}[b]{0.48\textwidth}
        \includegraphics[width=\textwidth]{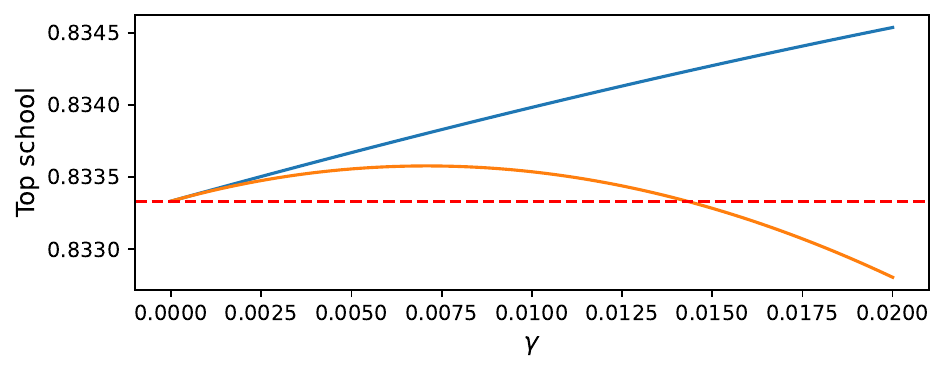}
        \caption{}
        \label{fig:gamma1k5lower}
    \end{subfigure}
    \hfill % This will insert a small space between the two subfigures
    \begin{subfigure}[b]{0.48\textwidth}
        \includegraphics[width=\textwidth]{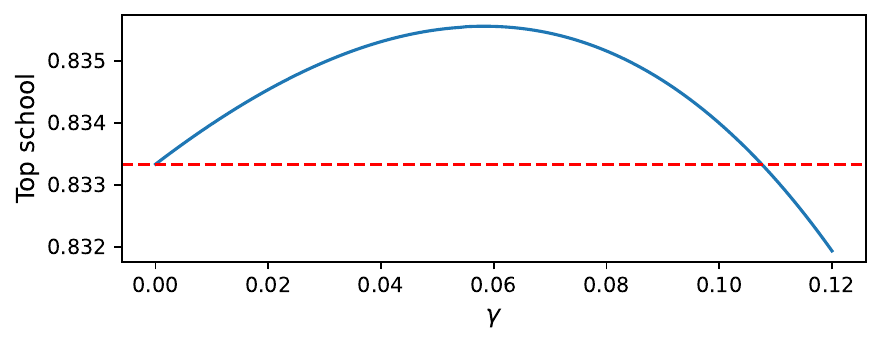}
        \caption{}
        \label{fig:gamma1k5}
    \end{subfigure}
    \caption{In both plots, the blue line is the value of $x_1$ in a $5$-school portfolio. The red dashed line at $5/6$ marks the top school an unbiased agent applies to. $(a)$ focuses on the range in which our analytical bound holds and plots the lower bound in orange. $(b)$ zooms out and shows the value of $x_1$ for a larger interval.}
    \label{fig:overshoot-k5}
\end{figure}

%% file: paper_sections/conclusion.tex
\section{Conclusion} \label{sec:conclusion}

In this paper we study a behavioral model based on reference points and loss aversion for the school application decisions made by students. 
We find that even a small amount of this bias leads to significant changes in how a student chooses where to apply, relative to the baseline of a student who does not experience these biases.
In particular, biased students in our model apply very sparsely to the most selective schools and devote an increasingly large fraction of their applications to less selective ones; but for small levels of bias, they also avoid the prospect of disappointing rejections by overshooting and apply to schools that are too selective.

There are a number of interesting directions for further work.  First, while we have used a standard and well-motivated model of the application process, with schools totally ordered by selectivity, it would be interesting to consider models where the outcomes at different schools are less strongly dependent, such as one finds in certain models of simultaneous search \cite{chade2006simultaneous}.
It would also be interesting to expand the functional form of the loss aversion used here; while manifesting the bias as an increased linear cost is standard, one could also ask how nonlinear costs might affect the outcome. 

% affects the school applications of students. We propose a simple model of this behavioral bias, and show how the student’s choices change in interesting and dramatic ways in a model where these reputation-based behavioral biases are taken into account. Our model is rich enough to cover a range of different ways that biased students cope with fear of rejection. We show that biased students apply very sparsely to highly selective schools, but at time overshoot rational students and apply to schools that are too selective. 

%% file: paper_sections/appendix.tex
\section{Missing Proofs from Section \ref{sec:approx}} \label{app-approx}
In this Section we prove:
\begin{lemma} 
For every $1 \leq i \leq k$, let $\bar \delta_i$ and $\underline \delta_i$ be such that: $\underline{\delta_i} \leq -(2x_i-3x_i^2) \leq \bar \delta_i$. The following holds:
\begin{align*}
\frac{k+1-i}{k+2-i} x_{i-1} + \frac{\gamma}{k+2-i} \sum_{j=1}^{k+1-i} j \cdot  \underline{\delta}_{k+1-j} \leq x_{i} \leq \frac{k+1-i}{k+2-i}  \cdot x_{i-1} + \frac{\gamma}{k+2-i} \sum_{j=1}^{k+1-i} j \cdot  \bar{\delta}_{k+1-j}.
\end{align*}
\end{lemma}
\begin{proof}
Consider Equation (\ref{eq-der-0}):
\begin{align*}
x_{i+1} &= 2x_i - x_{i-1} + \gamma(2x_i - 3x_i^2) \implies x_i = \frac{x_{i-1} + x_{i+1}}{2} - \frac{\gamma(2x_i - 3x_i^2)}{2}.
\end{align*}
% By rearranging we get that:
% \begin{align*}
% x_i = \frac{x_{i-1} + x_{i+1}}{2} - \frac{\gamma(2x_i - 3x_i^2)}{2}.
% \end{align*}
Using our assumptions we have that:
\begin{align*}
\frac{x_{i-1} + x_{i+1}}{2} - \gamma \cdot \frac{\underline{\delta}_i}{2 } \leq x_{i} \leq \frac{x_{i-1} + x_{i+1}}{2} + \gamma \cdot \frac{\bar{\delta}_i}{2}. \tag{$*$}
\end{align*}

We now use this to prove the lemma by backward induction. As the proofs of the upper bound and lower bound are the same, here we will only include the proof of the upper bound. So the base case is $i=k$. In this case, from ($*$) we have that:
\begin{align*}
x_{k} \leq \frac{x_{k-1}}{2} + \gamma\frac{\bar \delta_k}{2}.
\end{align*}
We now assume that the induction hypothesis holds for $i+1$ and prove it holds for $i$. Substituting the induction hypothesis into ($*$), we get that:
\begin{align*}
x_{i} \leq \frac{x_{i-1} + \frac{k-i}{k+1-i}  \cdot x_{i} + \frac{\gamma}{k+1-i} \sum_{j=1}^{k-i} j \cdot  \bar{\delta}_{k+1-j}}{2} + \gamma\frac{\delta_i}{2}.
\end{align*}
After rearranging, we get that:
\begin{align*}
x_{i} \leq \frac{x_{i-1}}{2} + \frac{k-i}{2(k+1-i)}  \cdot x_{i} + \frac{\gamma}{2(k+1-i)} \sum_{j=1}^{k-i} j \cdot  \bar{\delta}_{k+1-j} + \gamma\frac{\delta_i}{2}.
\end{align*}
This implies that:
\begin{align*}
\frac{k+2-i}{2(k+1-i)} x_{i} \leq \frac{x_{i-1}}{2} + \frac{\gamma}{2(k+1-i)} \sum_{j=1}^{k-i} j \cdot  \bar{\delta}_{k+1-j} + \gamma\frac{\delta_i}{2}.
\end{align*}
Hence,
\begin{align*}
x_{i} &\leq \frac{k+1-i}{k+2-i} x_i  + \frac{\gamma}{k+2-i} \sum_{j=1}^{k-i} j \cdot  \bar{\delta}_{k+1-j} + \gamma \frac{k+1-i}{k+2-i}
\bar{\delta}_i \\
&\leq \frac{k+1-i}{k+2-i}  \cdot x_{i-1} + \frac{\gamma}{k+2-i} \sum_{j=1}^{k+1-i} j \cdot  \bar{\delta}_{k+1-j}.
\end{align*}
\end{proof}

\section{Numerical computations of the biased student's strategy} \label{sec:tables}

Here, we include numerical values for the computational observations of the biased student's strategy from Section~\ref{sec:comp_obs}.

\begin{table}[htp]
\centering
\caption{Stabilization of expanding optimal portfolio for $\gamma = 1, 0.5, 0.1$.}
\label{table:freezing_gaps}
\begin{tabular}{cl}
\toprule
\multicolumn{2}{c}{$\gamma = 1$}             \\
$k$ & Schools                                               \\
\midrule
1 & 0.333                                                   \\
2 & 0.391, 0.106                                            \\
3 & 0.398, 0.117, 0.030                                     \\
4 & 0.398, 0.118, 0.032, 0.008                              \\
5 & 0.399, 0.118, 0.032, 0.008, 0.002                       \\
6 & 0.399, 0.118, 0.032, 0.008, 0.002, 0.005                \\ 
\midrule
\multicolumn{2}{c}{$\gamma = 0.5$}             \\
$k$ & Schools                                               \\
\midrule
1 & 0.422                                                   \\
2 & 0.588, 0.207                                            \\
3 & 0.610, 0.272, 0.095                                     \\
4 & 0.624, 0.288, 0.116, 0.039                              \\
5 & 0.626, 0.291, 0.119, 0.046, 0.015                       \\
6 & 0.627, 0.291, 0.120, 0.047, 0.017, 0.005                \\
7 & 0.627, 0.291, 0.120, 0.047, 0.018, 0.006, 0.002         \\
8 & 0.627, 0.291, 0.120, 0.047, 0.018, 0.006, 0.002, 0.0008 \\
\midrule
\multicolumn{2}{c}{$\gamma = 0.1$}             \\
        $k$ & Schools \\
\midrule
       1  & 0.486  \\
       2  & 0.653, 0.310 \\
       3  & 0.741, 0.465, 0.218 \\
       4  & 0.795, 0.560, 0.343, 0.159 \\
       5  & 0.833, 0.626, 0.425, 0.256, 0.118 \\
       6  & 0.862, 0.673, 0.483, 0.320, 0.190, 0.087 \\
       7  & 0.883, 0.708, 0.525, 0.364, 0.236, 0.139, 0.063  \\
       8  & 0.898, 0.734, 0.555, 0.394, 0.266, 0.170, 0.099, 0.045 \\
       9  & \makecell[tl]{0.908, 0.751, 0.574, 0.414, 0.285, 0.188, 0.119, 0.069,  0.031 } \\
       10  & \makecell[tl]{0.915, 0.761, 0.587, 0.426, 0.296, 0.198, 0.129, 0.081,  0.046, 0.021 } \\
       11  & \makecell[tl]{0.918, 0.767, 0.593, 0.432, 0.301, 0.204, 0.134, 0.086,  0.053, 0.031, 0.014 } \\
       12  & \makecell[tl]{0.920, 0.770, 0.596, 0.435, 0.304, 0.206, 0.137, 0.089,  0.057, 0.035, 0.020, 0.009 } \\
       13  & \makecell[tl]{0.920, 0.771, 0.598, 0.436, 0.305, 0.207, 0.138, 0.090,  0.058, 0.037, 0.023, 0.013, 0.006 } \\
       14  & \makecell[tl]{0.921, 0.772, 0.598, 0.437, 0.306, 0.208, 0.138, 0.090,  0.059, 0.038, 0.024, 0.014, 0.008, 0.003 } \\
       % 15  & \makecell[tl]{0.921, 0.772, 0.598, 0.437, 0.306, 0.208, 0.138, 0.091,  0.059, 0.038, 0.024, 0.015, 0.009, 0.005, 0.002 } \\
  \bottomrule                                 
\end{tabular}
\end{table}

\begin{table}[htp]
\caption{Convergence of true expected payoff from $k=1$ to $100$, for selected values of $\gamma$.}
\label{table:payoff_converge}
\centering
\begin{tabular}{crrrrrr}
\toprule
$k$   & Unbiased ($\frac{k}{2(k+1)}$) & $\gamma = 0.01$ & $\gamma = 0.05$ & $\gamma = 0.1$ & $\gamma = 0.2$ & $\gamma=0.5$ \\
\midrule
1   & 0.25     & 0.249998      & 0.249958      & 0.249827     & 0.249242     & 0.244016   \\
2   & 0.333333 & 0.333329      & 0.333237      & 0.332936     & 0.331625     & 0.319467   \\
3   & 0.375    & 0.374991      & 0.374787      & 0.374121     & 0.371246     & 0.346748   \\
4   & 0.4      & 0.399984      & 0.399586      & 0.398285     & 0.392758     & 0.354428   \\
5   & 0.416666 & 0.416638      & 0.415945      & 0.413675     & 0.404421     & 0.356059   \\
6   & 0.428571 & 0.428527      & 0.427412      & 0.423796     & 0.410208     & 0.356349   \\
7   & 0.4375   & 0.437433      & 0.435752      & 0.430424     & 0.412700     & 0.356400   \\
8   & 0.444444 & 0.444349      & 0.441937      & 0.434621     & 0.413637     & 0.356407   \\
9   & 0.45     & 0.449868      & 0.446553      & 0.437132     & 0.413961     & 0.356407   \\
10  & 0.454545 & 0.454370      & 0.449979      & 0.438526     & 0.414065     & 0.356409   \\
25  & 0.480769 & 0.478386      & 0.457684      & 0.439903     & 0.414113     & 0.356409   \\
50  & 0.490196 & 0.481122      & 0.457685      & 0.439903     & 0.414113     & 0.356409   \\
75  & 0.493421 & 0.481131      & 0.457685      & 0.439903     & 0.414113     & 0.356409   \\
100 & 0.495049 & 0.481131      & 0.457685      & 0.439903     & 0.414113     & 0.356409  \\ 
\bottomrule
\end{tabular}
\end{table}